\documentclass{gen-j-l}
\usepackage [utf8]{inputenc}
\usepackage[T1]{fontenc}
\usepackage{amssymb}
\usepackage{pstricks}
\usepackage{pst-plot}
\usepackage{eso-pic}
\usepackage{graphicx}
\usepackage{tikz}
\usepackage{listings}
\usepackage{enumerate}
\usepackage[cmtip,all]{xy}
\copyrightinfo{}{}

\newtheorem{theo}{Theorem}[subsection]

\newtheorem{theoa}{Theorem}
\newtheorem{prop}[theo]{Proposition}

\newtheorem{cora}{Corollary}[theoa]
\newtheorem{defi}[theo]{Definition}

\newtheorem{rema}[theo]{Remark}

\newcommand{\Hol}{\mathrm{Hol}}

\newcommand{\Op}{\mathrm{Op}}

\newcommand{\id}{\mathrm{id}}

\newcommand{\diag}{\mathrm{diag}}

\begin{document}

\title[Quantizations of the torus]{Berezin-Toeplitz quantization and complex Weyl quantization of the torus $ \mathbb{T}^2$.}

\author{Ophélie Rouby}
\address{Grupo de Física Matemática, Faculdade de Ciências, Universidade de Lisboa, 1749-016 Lisboa, Portugal}
\curraddr{}
\email{o.rouby@laposte.net}
\thanks{}

\subjclass[1991]{}

\date{}

\dedicatory{}

\begin{abstract}
In this paper, we give a correspondence between the Berezin-Toeplitz and the complex Weyl quantizations of the torus $ \mathbb{T}^2$. To achieve this, we use the correspondence between the Berezin-Toeplitz and the complex Weyl quantizations of the complex plane and a relation between the Berezin-Toeplitz quantization of a periodic symbol on the real phase space $ \mathbb{R}^2$ and the Berezin-Toeplitz quantization of a symbol on the torus $ \mathbb{T}^2$. 
\end{abstract}

\maketitle
\setcounter{tocdepth}{1}

\section*{Introduction}
The object of this paper is to construct a new semi-classical quantization of the torus $ \mathbb{T}^2$ by adapting Sjöstrand's complex Weyl quantization of $ \mathbb{R}^2$ and to give the correspondence between this quantization and the well-known Berezin-Toeplitz quantization of $ \mathbb{T}^2$.
When the phase space is $ \mathbb{R}^{2n}$, the pseudo-differential Weyl quantization allows us to relate a classical system to a quantum one through the symbol map; thus pseudo-differential operators have become an important tool in quantum mechanics. On the mathematical side, these operators have been introduced in the mid-sixties by André Unterberger and Juliane Bokobza \cite{MR0176360} and in parallel by Joseph Kohn and Louis Nirenberg \cite{MR0176362} and have been investigated by Lars Hörmander \cite{MR0180740,MR0233064,MR0383152}. They allow to study physical systems in positions and momenta. On the other hand, Berezin-Toeplitz operators have been introduced by Feliks Berezin \cite{MR0411452} and investigated by Louis Boutet de Monvel and Victor Guillemin \cite{MR620794} as a generalization of Toeplitz matrices. The study of these operators has been motivated by the fact that pseudo-differential operators take into account only phases spaces that can be written as cotangent spaces, whereas in mechanics, there are physical observables like spin that naturally lives on other types of phases spaces, like compact Kähler manifolds, which can be quantized in the Berezin-Toeplitz way. In fact, it was realized recently that the Berezin-Toeplitz quantization applies to even more general symplectic manifolds, and thus has become a tool of choice for applications of symplectic geometry and topology, see \cite{2016arXiv160905395C}. \\
\\
In this paper, we give a relation between the Berezin-Toeplitz quantization of the torus, studied for instance by David Borthwick and Alejandro Uribe in \cite{MR2014161} and the complex Weyl quantization of the torus, which we introduce  as a variation of Sjöstrand's quantization of $ \mathbb{R}^2$. The complex Weyl quantization of $ \mathbb{R}^2$ has been investigated by Johannes Sj{\"o}strand in \cite{SJ}, then by Anders Melin and Johannes Sj{\"o}strand in \cite{MR1957486,MR2003421}, also by Michael Hitrik and Johannes Sj{\"o}strand in \cite{MR2036816} and in their mini-courses \cite{HJ} and by Michael Hitrik, Johannes Sj{\"o}strand and San V\~u Ng\d{o}c in \cite{MR2288739}. This quantization of the real plane $ \mathbb{R}^2$ allows to study pseudo-differential operators with complex symbols, and therefore is particularly useful for problems involving non self-adjoint operators or quantum resonances. It is defined by a contour integral over an $IR$-manifold ($I$-Lagrangian and $R$-symplectic) which plays the role of the phase space. Here, we define an analogue of this notion in the torus case. \\
If we consider the complex plane as a phase space, there exists a correspondence between the complex Weyl and the Berezin-Toeplitz quantizations (this correspondence uses a variant of Bargmann's transform and can be found, for instance, in the book \cite[chapter 13]{MR2952218} of Maciej Zworski); using this result, we are able to obtain Bohr-Sommerfeld quantization conditions for non-selfadjoint perturbations of self-adjoint Berezin-Toeplitz operators of the complex plane $ \mathbb{C}$ by first proving the result in the case of pseudo-differential operators (see \cite{lapin}). Therefore, we expect that this new complex quantization of $ \mathbb{T}^2$, together with its relationship to the Berezin-Toeplitz quantization, will be crucial in obtaining precise eigenvalue asymptotics of non-selfadjoint Berezin-Toeplitz operators on the torus.
\\

\noindent \textit{Structure of the paper:}
\begin{itemize}
\item in Section \ref{section_resultat}, we state our result;
\item in Section \ref{section_demo}, we give the proof of our result which is divided into three parts, the first one consists in recalling the Berezin-Toeplitz quantization of the torus, the second one in introducing the complex Weyl quantization of the torus and the last one in relating these two quantizations.
\end{itemize}
$ $  

\noindent \textbf{Acknowledgements.} The author would like to thank both San V\~u Ng\d{o}c and Laurent Charles for their support and guidance. Funding was provided by the Universit\'e de Rennes 1 and the Centre Henri Lebesgue.

\section{Result} \label{section_resultat}
\subsection{Context} \label{subsection_context}
In this section, we recall the definition of the Berezin-Toeplitz quantization of a symbol on the torus $ \mathbb{T}^2$ (see for example \cite{MR3349834}) and we give a definition of the complex Weyl quantization of a symbol on the torus. Let $ 0 < \hbar \leq 1$ be the semi-classical parameter. By convention the Weyl quantization involves the semi-classical parameter $ \hbar$, contrary to the Berezin-Toeplitz quantization which involves the inverse of this parameter, denoted by $k$. In the whole paper, we will use these two parameters. \\

\noindent \texttt{Notation:} let $k$ be an integer greater than $1$. Let $u$ and  $v$ be complex numbers of modulus $1$.
\begin{itemize}
\item If $z \in \mathbb{C}$, we denote by $z = (p, q) \in \mathbb{R}^2$ or $z = p + iq$ via the identification of $ \mathbb{C}$ with $ \mathbb{R}^2$.
\item $ \mathbb{T}^2$ denotes the torus $ (\mathbb{R}/ 2 \pi \mathbb{Z}) \times ( \mathbb{R}/ \mathbb{Z})$.
\item $ \mathcal{G}_k$ is the space of measurable functions $g$ such that:
$$ \int_0^{2 \pi} \! \! \! \int_0^1 \left| g(p, q) \right|^2 e^{-k q^2} dp dq < + \infty,$$
which are invariant under the action of the Heisenberg group (for more details, see Subsection \ref{subsection_quantif_toeplitz}), \textit{i.e.} for all $(p, q) \in \mathbb{R}^2$, we have:
$$ g(p+2 \pi, q) = u^k g(p, q)\quad \text{and} \quad g(p, q+1) = v^k e^{-i(p+iq)k+k/2} g(p, q).$$
\item $ \mathcal{H}_k$ is the space of holomorphic functions in $ \mathcal{G}_k$, \textit{i.e.}:
$$ \mathcal{H}_k = \left\lbrace g \in \Hol( \mathbb{C}); \quad g(p+2 \pi, q) = u^k g(p, q),\quad g(p, q+1) = v^k e^{-i(p+iq)k+k/2} g(p, q) \right\rbrace.$$
\item $\Pi_k$ is the orthogonal projection of the space $ \mathcal{G}_k$ (equipped with the weighted $L^2$-scalar product on $[0, 2 \pi] \times [0, 1]$) on the space $\mathcal{H}_k$.
\end{itemize} 

\begin{rema} $ $
\begin{itemize}
\item The spaces $ \mathcal{G}_k$ and $ \mathcal{H}_k$ depend on the complex numbers $u$ and $v$. 
\item In \cite{MR2014161}, they consider the torus $ \mathbb{T}^2 = \mathbb{R}^2 / \mathbb{Z}^2$ and they choose an other quantization which leads to an other space of holomorphic functions, also called $ \mathcal{H}_k$, defined as follows:
$$ \mathcal{H}_k = \left\lbrace g \in \Hol( \mathbb{C}); \quad \forall (m, n) \in \mathbb{Z}^2, g(z+m+in) = (-1)^{kmn} e^{k \pi ( z(m-in) + (1/2)(m^2 + n^2))} g(z) \right\rbrace.$$
\end{itemize}
\end{rema} 

\begin{defi}[Asymptotic expansion] \label{defi_asymp_expan}
Let $f_k \in \mathcal{C}^{ \infty}( \mathbb{R}^2)$. We say that $f_k$ admits an asymptotic expansion in powers of $1/k$ for the $ \mathcal{C}^{ \infty}$-topology of the following form:
$$ f_k(x,y) \sim \sum_{l \geq 0} k^{-l} f_l(x,y) ,$$
if:
\begin{enumerate}
\item $ \forall l \in \mathbb{N}$, $f_l \in \mathcal{C}^{ \infty}( \mathbb{R}^2)$;
\item $ \forall L \in \mathbb{N}^*$, $ \forall (x, y) \in \mathbb{R}^2$, $ \exists C > 0$ such that:
$$ \left| f_k(x,y) - \sum_{l=0}^{L-1} k^{-l} f_l(x,y) \right| \leq C k^{-L} \quad \text{for large enough $k$}.$$
\end{enumerate}
We denote by $ \mathcal{C}^{ \infty}_k( \mathbb{R}^2)$ the space of such functions.
\end{defi}

\begin{defi}[Berezin-Toeplitz quantization of the torus] \label{defi_quantif_BT_tore}
Let $f_k \in \mathcal{C}^{ \infty}_k(\mathbb{R}^2)$ be a function such that, for $(x,y) \in \mathbb{R}^2$, we have:
$$ f_k(x+ 2 \pi, y) = f_k(x, y) = f_k(x, y+1 ) .$$
Define the Berezin-Toeplitz quantization of the function $f_k$ by the sequence of operators $ T_{f_k} := (T_k)_{k \geq 1}$ where, for $k \geq 1$, the operator $T_k$ is given by:
$$ T_k = \Pi_k M_{f_k} \Pi_k : \mathcal{H}_k \longrightarrow \mathcal{H}_k,$$
where $M_{f_k} : \mathcal{G}_k \longrightarrow \mathcal{G}_k$ is the multiplication operator by the function $f_k$. \\
We call $f_k$ the symbol of the Berezin-Toeplitz operator $T_{f_k}$.
\end{defi}

Now, we define the complex Weyl quantization of a symbol on the torus. We will explain in details in Subsection \ref{subsection_complex_weyl_quantization} why we consider such a notion. First, we introduce some notations. \\

\noindent \texttt{Notation:} let $ \Phi_1$ be the strictly subharmonic quadratic form defined by the following formula for $z \in \mathbb{C}$:
$$ \Phi_1(z) = \dfrac{1}{2} \Im(z)^2 .$$
\begin{itemize}
\item $ \Lambda_{ \Phi_1}$ denotes the following set:
$$ \Lambda_{ \Phi_1} = \left\lbrace \left(z, \dfrac{2}{i} \dfrac{\partial \Phi_1}{\partial z}(z) \right); z \in \mathbb{C} \right\rbrace = \lbrace (z, - \Im(z)); z \in \mathbb{C} \rbrace \simeq \mathbb{C}.$$
\item $L(dz)$ denotes the Lebesgue measure on $ \mathbb{C}$, \textit{i.e.} $L(dz)= \dfrac{i}{2} dz \wedge d \overline{z}$.
\item $L^2_{\hbar}( \mathbb{C}, \Phi_1) := L^2( \mathbb{C}, e^{-2 \Phi_1(z)/ \hbar} L( d z))$ is the set of measurable functions $f$ such that: 
$$ \int_{ \mathbb{C}} |f(z)|^2 e^{-2 \Phi_1(z)/ \hbar} L(dz) < + \infty.$$
\item $H_{ \hbar}( \mathbb{C}, \Phi_1) := \Hol( \mathbb{C}) \cap L^2_{\hbar}( \mathbb{C}, \Phi_1)$ is the set of holomorphic functions in the space $L^2_{ \hbar}( \mathbb{C}, \Phi_1)$.
\item $ \mathcal{C}^{ \infty}_{ \hbar}( \Lambda_{ \Phi_1})$ denotes the set of smooth functions on $ \Lambda_{ \Phi_1}$ admitting an asymptotic expansion in powers of $ \hbar$ for the $ \mathcal{C}^{ \infty}$-topology in the sense of Definition \ref{defi_asymp_expan} (by replacing $1/k$ by $ \hbar$ and $ \mathbb{R}^2$ by $ \Lambda_{ \Phi_1}$).
\end{itemize}

\begin{rema}
There are several definitions of the Bargmann transform. Here we chose the weight function $ \Phi_1 (z) = \dfrac{1}{2} \Im(z)^2$ instead of $ |z|^2$ because it is well-adapted to the analysis of the torus.
\end{rema}

\begin{defi}[Complex Weyl quantization of the torus $ \mathbb{T}^2$ (see Definitions \ref{defi_weyl_complex_tore_1} and \ref{defi_weyl_complex_tore_2})]  \label{defi_quantif_complex_tore}
Let $b_{ \hbar} \in \mathcal{C}^{ \infty}_{ \hbar}( \Lambda_{ \Phi_1})$ be a function such that, for $(z, w) \in \Lambda_{ \Phi_1}$, we have: 
$$b_{ \hbar}(z + 2 \pi, w) = b_{ \hbar}(z, w) = b_{ \hbar}(z+i, w-1).$$
Define the complex Weyl quantization of the function $b_{ \hbar}$, denoted by $ \Op^w_{ \Phi_1}(b_{ \hbar})$, by the following formula, for $u \in H_{ \hbar}( \mathbb{C}, \Phi_1)$:
$$ \Op_{ \Phi_1}^w(b_{ \hbar}) u(z) = \dfrac{1}{2 \pi \hbar} \int \! \! \! \int_{ \Gamma(z)} e^{(i/ \hbar)(z-w) \zeta} b_{ \hbar} \left( \dfrac{z+w}{2}, \zeta \right) u(w) dw d \zeta,$$
where the contour integral is the following:
$$ \Gamma(z) = \left\lbrace (w, \zeta) \in \mathbb{C}^2; \zeta =  \dfrac{2}{i} \dfrac{\partial \Phi_1}{\partial z} \left( \dfrac{z+w}{2} \right) = - \Im \left( \dfrac{z+w}{2} \right) \right\rbrace.$$
We call $b_{ \hbar}$ the symbol of the pseudo-differential operator $ \Op^w_{ \Phi_1}(b_{ \hbar})$.
\end{defi}

We will show that for $b_{ \hbar} \in \mathcal{C}^{ \infty}_{ \hbar}( \Lambda_{ \Phi_1})$ satisfying the hypotheses of Definition \ref{defi_quantif_complex_tore}, the complex Weyl quantization defines an operator $ \Op^w_{ \Phi_1}(b_{ \hbar})$ which sends the space of holomorphic functions $ \mathcal{H}_k$ on itself (see Proposition \ref{prop_action_quantif_weyl_complexe_sur_SS'(C)}). Therefore, the Berezin-Toeplitz and the complex Weyl quantizations give rise to operators acting on the space of holomorphic functions $ \mathcal{H}_k$. 

\subsection{Main result} 

\begin{theoa} \label{theoA}
Let $ f_k \in \mathcal{C}^{ \infty}_k( \mathbb{R}^2)$ be a function such that, for $(x, y) \in \mathbb{R}^2$, we have:
$$ f_k(x + 2 \pi, y ) = f_k(x, y) = f_k(x, y+1) .$$
Let $T_{f_k} = ( T_k)_{k \geq 1}$ be the Berezin-Toeplitz operator of symbol $ f_k$. Then, for $k \geq 1$, we have:
$$ T_k = \Op^w_{ \Phi_1}(b_{ \hbar}) + \mathcal{O}(k^{- \infty}) \quad \text{on $ \mathcal{H}_k$},$$
where $b_{\hbar} \in \mathcal{C}^{ \infty}_{ \hbar}( \Lambda_{\Phi_1})$ is given by the following formula, for $z \in \Lambda_{ \Phi_1} \simeq \mathbb{C}$:
$$ b_{ \hbar}(z) = \exp \left( \dfrac{1}{k} \partial_z \partial_{ \overline{z}} \right) (f_k(z)).$$
This formula means that $b_{ \hbar}$ is the solution at time $1$ of the following ordinary differential equation:
$$
\begin{cases}
\partial_t b_{ \hbar}(t, z) = \dfrac{1}{k} \partial_z \partial_{ \overline{z}} \left( b_{ \hbar}(t, z) \right), \\
b_{ \hbar}(0,z) = f_k(z).
\end{cases}
$$
Besides, $b_{ \hbar}$ satisfies the following periodicity conditions, for $(z,w) \in \Lambda_{ \Phi_1}$:
$$ b_{ \hbar}(z+ 2 \pi, w) = b_{ \hbar}(z, w) = b_{ \hbar}(z+i, w-1).$$
\end{theoa}

\begin{rema}
This result is analogous to Proposition \ref{prop_Toeplitz=pseudo_H(Phi1)} (see for example \cite[Chapter 13]{MR2952218}) which relates the Berezin-Toeplitz and the complex Weyl quantizations of the complex plane. The important difference here is that the phase space is the torus.
\end{rema}

\begin{rema}
As a corollary of this result, we can establish a connection between the Berezin-Toeplitz and the \emph{classical} Weyl quantizations of the torus (see Corollary \ref{coro}).
\end{rema}

\section{Proof} \label{section_demo}

\noindent The structure of the proof is organized as follows:
\begin{itemize}
\item in Subsection \ref{subsection_quantif_toeplitz}, we recall the Berezin-Toeplitz quantization of the torus;
\item in Subsection \ref{subsection_complex_weyl_quantization}, we introduce the complex Weyl quantization of the torus;
\item in Subsection \ref{subsection_links_between_quantizations}, we relate the Berezin-Toeplitz quantization of the torus to the complex Weyl quantization of the torus.
\end{itemize}

\subsection{Berezin-Toeplitz quantization of the torus $ \mathbb{T}^2$} \label{subsection_quantif_toeplitz}
In this paragraph, we recall the geometric quantization of the torus (see for example the article \cite{MR3349834} of Laurent Charles and Julien March\'e). \\

Consider the real plane $ \mathbb{R}^2$ endowed with the euclidean metric, its canonical complex structure and with the symplectic form $ \omega = dp \wedge dq$. Let $L_{ \mathbb{R}^2} = \mathbb{R}^2 \times \mathbb{C}$ be the trivial complex line bundle endowed with the constant metric and the connection $ \nabla = d + \dfrac{1}{i} \alpha$ where $ \alpha$ is the $1$-form given by:
$$ \alpha = \dfrac{1}{2} \left( p dq - q dp \right).$$
The holomorphic sections of $L_{ \mathbb{R}^2}$ are the sections $f$ satisfying the following condition: 
$$ \nabla_{ \overline{z}} f = \dfrac{\partial}{\partial \overline{z}} f + \dfrac{1}{4} z f = 0 .$$

We are interested in the holomorphic sections of the torus $ \mathbb{T}^2 = ( \mathbb{R}/ 2 \pi \mathbb{Z}) \times ( \mathbb{R}/ \mathbb{Z})$. Let $x = 2 \pi \dfrac{\partial}{\partial p}$, \textit{i.e.} if we denote by $t_{x}$ the translation of vector $x$, it is defined by the following formula:
\begin{align*}
t_{x} : \mathbb{R}^2 & \longrightarrow \mathbb{R}^2 \\
(p,q) & \longrightarrow (p+2 \pi, q).
\end{align*}
And let $y = \dfrac{\partial}{\partial q}$ which corresponds to the translation $t_y$ given by:
\begin{align*}
t_y: \mathbb{R}^2 & \longrightarrow \mathbb{R}^2 \\
(p, q) & \longmapsto (p, q+1).
\end{align*}
Note that the $ \omega$ volume of the fundamental domain of the lattice is $ 2 \pi$. \\
Let $k \geq 1$, the Heisenberg group at level $k$ is $ \mathbb{R}^2 \times U(1)$ with the product:
$$ (x, u) \cdot (y, v) = \left( y+x, u v e^{(ik/2) \omega(x, y)} \right),$$
for $(x, u), (y, v) \in \mathbb{R}^2 \times U(1)$ (where $U(1)$ denotes the set of complex numbers of modulus one). This formula defines an action of the Heisenberg group on the bundle $L_{ \mathbb{T}^2}^{ \otimes k}$ endowed with the product measure. We identify the space of square integrable sections of $L_{ \mathbb{T}^2}^{ \otimes k}$ which are invariant under the action of the Heisenberg group with the space $ \mathcal{G}_k$ (defined in Subsection \ref{subsection_context}). In fact, if $ \psi$ denotes such a section, we associate to it a function $g \in \mathcal{G}_k$ using the following application:
\begin{align*}
L^2( \mathbb{T}^2, L^{ \otimes k}_{ \mathbb{T}^2}) & \longrightarrow \mathcal{G}_k \\
\psi & \longmapsto g( \tilde{x}),
\end{align*}
where $ \tilde{x} \in \mathbb{R}^2$ and $ \tilde{x} = x_0 + (n_1,n_2)$ with $ x_0 \in [0, 2 \pi] \times [0,1]$, $(n_1, n_2) \in \mathbb{Z}^2$ and where:
$$ ( \tilde{x}, g( \tilde{x})) = ((n_1, n_2),1)  \cdot (x_0, \psi(x_0)) .$$
 Similarly, we identify the space of holomorphic sections of $L_{ \mathbb{T}^2}^{ \otimes k}$ with the following Hilbert space:
$$ \mathcal{H}_k =  \left\lbrace g \in \Hol( \mathbb{C}); \quad g(p+2 \pi, q) = u^k g(p, q), \quad g(p, q+1) = v^k e^{-i(p+iq)k+k/2} g(p, q) \right\rbrace,$$
endowed with the $L^2$-weighted scalar product on $[0, 2 \pi] \times [0,1]$. The complex numbers $u$ and $v$ are called Floquet indices. The Hilbert space $ \mathcal{H}_k$ admits an orthogonal basis, given 
for $l \in \lbrace 0, \ldots, k-1 \rbrace$, by the functions $e_l$ which are defined, for $z \in \mathbb{C}$, as follows:
\begin{equation} \label{eq_base_el}
e_l(z) = u^{kz/(2 \pi)} \sum_{j \in \mathbb{Z}} \left( v^{-k} e^{-l-kj/2} u^{ik/(2 \pi)} \right)^j e^{i(l+jk)z} .
\end{equation}

\subsection{Complex Weyl quantization of the torus $ \mathbb{T}^2$} \label{subsection_complex_weyl_quantization}
In this paragraph, we introduce the notion of complex Weyl quantization of the torus which, to our knowledge, is new. To do so, we follow these three steps:
\begin{enumerate}
\item[1.] we recall the definition of the classical Weyl quantization of the torus;
\item[2.] we recall the definition of the semi-classical Bargmann transform and we look at some of its properties;
\item[3.] we introduce the complex Weyl quantization as the range of the classical Weyl quantization by the Bargmann transform.
\end{enumerate}

\subsubsection{Classical Weyl quantization of the torus}
The classical Weyl quantization of a symbol on the torus has been studied, for example, by Monique Combescure and Didier Robert in the book \cite[Chapter 6]{MR2952171}. We need to introduce the following notation.\\

\noindent \texttt{Notation:}
\begin{itemize}
\item $ \mathcal{S}( \mathbb{R})$ denotes the Schwartz space, \textit{i.e.}:
$$  \mathcal{S}( \mathbb{R}) = \left\lbrace \phi \in \mathcal{C}^{ \infty}( \mathbb{R}); \|\phi \|_{\alpha, \beta} := \sup_{ x \in \mathbb{R}} | x^{ \alpha} \partial^{ \beta}_x \phi(x) | < + \infty, \forall \alpha, \beta \in \mathbb{N} \right\rbrace; $$
\item for $ \phi \in \mathcal{S}( \mathbb{R})$, $ \mathcal{F}_{ \hbar} \phi$ denotes the semi-classical Fourier transform of the function $ \phi$ and it is defined by the following equality:
$$ \mathcal{F}_{ \hbar} \phi(\xi) = \int_{ \mathbb{R}} e^{-(i/ \hbar) x \xi} \phi(x) dx$$
this transform is an isomorphism of the Schwartz space and its inverse is given by:
$$ \mathcal{F}_{ \hbar}^{-1} \phi(x) = \dfrac{1}{2 \pi \hbar} \int_{ \mathbb{R}} e^{(i/ \hbar) x \xi} \phi( \xi) d \xi ;$$
\item $ \mathcal{S}'( \mathbb{R})$ denotes the space of tempered distributions, it is the dual of the Schwartz space $ \mathcal{S}( \mathbb{R})$, \textit{i.e.} it is the space of continuous linear functionals on $ \mathcal{S}( \mathbb{R})$;
\item $ \langle \cdot, \cdot \rangle_{ \mathcal{S}', \mathcal{S}}$ denotes the duality bracket between $ \mathcal{S}'( \mathbb{R})$ and $ \mathcal{S}( \mathbb{R})$;
\item for $ \psi \in \mathcal{S}'( \mathbb{R})$, $ \mathcal{F}_{ \hbar} \psi$ denotes the semi-classical Fourier transform of a tempered distribution and it is defined by the following equality, for $ \phi \in \mathcal{S}( \mathbb{R})$:
$$ \langle \mathcal{F}_{ \hbar} \psi, \phi \rangle_{ \mathcal{S}', \mathcal{S}} = \langle \psi, \mathcal{F}_{ \hbar} \phi \rangle_{ \mathcal{S}', \mathcal{S}} ;$$
\item for $a \in \mathbb{R}$, we denote by $\tau_a$ the translation of vector $a$ defined as follows:
\begin{align*}
\tau_a : \mathbb{R} & \longrightarrow \mathbb{R} \\
x & \longmapsto x+a,
\end{align*}
recall that the translation 
of a tempered distribution $ \psi \in \mathcal{S}'( \mathbb{R})$ is defined as follows, for $ \phi \in \mathcal{S}( \mathbb{R})$:
$$ \langle \tau_a \psi, \phi \rangle_{ \mathcal{S}', \mathcal{S}} = \langle \psi, \tau_{-a} \phi \rangle_{ \mathcal{= \mathbb{R} / \mathbb{Z}^2S}', \mathcal{S}} $$
the distribution $ \psi$ is called $a$-periodic if $\tau_a \psi = \psi$, in this case, $ \psi$ can be written as a convergent Fourier series in $ \mathcal{D}'( \mathbb{R})$ (see for example the book of Jean-Michel Bony \cite{Bony}):
$$ \psi = \sum_{l \in \mathbb{Z}} \psi_l e^{ilt 2 \pi/a} ,$$
where the sequence $( \psi_l)_{l \in \mathbb{Z}}$ is such that, there exists an integer $N \geq 0$ such that:
$$ | \psi_l| \leq C (1 + |l|)^N \quad \forall l \in \mathbb{Z}.$$
\end{itemize}
$ $ \\

Recall now the definition of the subspace of tempered distributions that corresponds to the natural space on which pseudo-differential operators of the torus act (see \cite[Chapter 6]{MR2952171}). For $k \geq 1$ and for $u, v \in U(1)$, we consider the following space:
$$ \mathcal{L}_k = \left\lbrace \psi \in \mathcal{S}'( \mathbb{R}); \quad \tau_{2 \pi} \psi  = u^k \psi, \quad \tau_1 \mathcal{F}_{ \hbar}( \psi) = v^{-k} \mathcal{F}_{ \hbar}( \psi)  \right\rbrace .$$

\begin{rema} $ $ 
\begin{itemize}
\item The definition of the space $ \mathcal{L}_k$ involves two complex numbers $u$ and $v$. We will see that they correspond to the Floquet indices seen in the definition of the space $ \mathcal{H}_k$.
\item In \cite{MR2679813}, they consider the torus $ \mathbb{T}^2 = \mathbb{R}^2/ \mathbb{Z}^2$ and they choose $u = v= 1$.
\end{itemize}

\end{rema}

This space admits a basis (see for example \cite[Chapter 6]{MR2952171}), given for $l \in \lbrace 0, \ldots, k-1 \rbrace$, by the distributions $ \epsilon_l$ which are defined as follows:
\begin{equation} \label{eq_base_epsilon_l}
\epsilon_l = u^{kt/(2 \pi)} \sum_{j \in \mathbb{Z}} \left( v^{-k} \right)^j e^{i(l+jk)t} .
\end{equation}
We consider the structure of Hilbert space such that the family $( \epsilon_l)_{l \in \mathbb{Z}}$ is an orthonormal basis of the space $ \mathcal{L}_k$. Recall now two different definitions of the Weyl quantization of a symbol on the torus $ \mathbb{T}^2$. In the whole paper $S( \mathbb{R}^2)$ denotes the following class of symbols on $ \mathbb{R}^2$:
$$ S( \mathbb{R}^2) = \left\lbrace a \in \mathcal{C}^{ \infty}( \mathbb{R}^2); \forall \alpha \in \mathbb{N}^2, \text{there exists a constant $C_{ \alpha} > 0$ such that:} \, \left| \partial^{ \alpha} a \right| \leq C_{ \alpha} \right\rbrace.$$

\begin{rema} \label{rema_symbole_perio_dans_S(R2)}
Let $a_{ \hbar} \in \mathcal{C}^{ \infty}_{ \hbar}( \mathbb{R}^2)$ be a function such that, for all $(x, y) \in \mathbb{R}^2$, we have:
$$ a_{ \hbar}(x + 2 \pi, y) = a_{ \hbar}(x, y) = a_{ \hbar}(x, y+1).$$
Then the function $a_{ \hbar}$ belongs to the class of symbols $S( \mathbb{R}^2)$.
\end{rema}

\begin{defi}[First definition of the Weyl quantization of the torus] \label{defi_quantif_weyl_tore_integrale}
Let $a_{ \hbar } \in \mathcal{C}^{ \infty}_{ \hbar}(\mathbb{R}^2)$ be a function such that, for all $(x, y) \in \mathbb{R}^2$, we have:
$$ a_{ \hbar}(x+2 \pi, y) = a_{ \hbar}(x, y) = a_{ \hbar}(x, y+1) .$$
Define the Weyl quantization of the symbol $a_{ \hbar}$, denoted by $ \Op^w(a_{ \hbar})(x, \hbar D_x)$, by the following integral formula, for $u \in \mathcal{S}( \mathbb{R})$:
$$ \Op^w(a_{ \hbar})(x, \hbar D_x) u(x) = \dfrac{1}{2 \pi \hbar}
 \int_{ \mathbb{R}} \int_{ \mathbb{R}} e^{i(x-y) \xi/ \hbar} a _{ \hbar}\left( \dfrac{x+y}{2}, \xi \right) u(y) dy d \xi .$$
We call $a_{ \hbar}$ the symbol of the pseudo-differential operator $ \Op^w(a_{ \hbar})(x, \hbar D_x)$. 
\end{defi}

Recall that if $a_{ \hbar} \in S( \mathbb{R}^2)$, then (see for example the book of Maciej Zworski \cite[Chapter 3]{MR2952218}):
\begin{enumerate}
\item[1.] $ \Op^w(a_{ \hbar})(x, \hbar D_x): \mathcal{S}( \mathbb{R}) \longrightarrow \mathcal{S}( \mathbb{R})$;
\item[2.] $ \Op^w(a_{ \hbar})(x, \hbar D_x): \mathcal{S}'( \mathbb{R}) \longrightarrow \mathcal{S}'( \mathbb{R})$;
\end{enumerate}
are continuous linear transformations and the action of $ \Op^w(a_{ \hbar})$ on $ \mathcal{S}'( \mathbb{R})$ is defined, for $\psi \in \mathcal{S}'( \mathbb{R})$ and $\phi \in \mathcal{S}( \mathbb{R})$, by:
\begin{equation} \label{equa_action_operateur_pseudo_sur_S'(R)}
\langle \Op^w(a_{ \hbar}) \psi, \phi \rangle_{ \mathcal{S}', \mathcal{S}} = \langle \psi, \Op^w( \tilde{a}_{ \hbar}) \phi \rangle_{ \mathcal{S}', \mathcal{S}},
\end{equation}
where, for $(x, y) \in \mathbb{R}^2$, $ \tilde{a}_{ \hbar}(x, y) := a_{ \hbar}(x, - y) \in S(\mathbb{R}^2)$. This property allows to easily prove the following proposition (see \cite{MR2679813}).

\begin{prop}
Let $a_{ \hbar} \in \mathcal{C}^{ \infty}_{ \hbar}( \mathbb{R}^2)$ be a function such that, for all $(x, y) \in \mathbb{R}^2$, we have:
$$ a_{ \hbar }(x + 2 \pi,y) = a_{ \hbar}(x, y) = a_{ \hbar}(x, y+1) .$$
Then, if $ \hbar = \dfrac{1}{k}$ for $k \geq 1$, we have: $ \Op^w(a_{ \hbar})(x, \hbar D_x): \mathcal{L}_k \longrightarrow \mathcal{L}_k$.
\end{prop}

Since we consider a symbol $a_{ \hbar} \in \mathcal{C}^{ \infty}_{ \hbar}( \mathbb{R}^2)$ which is periodic, we can rewrite it as a Fourier series, for all $(x, y) \in \mathbb{R}^2$:
\begin{equation} \label{equation_ecriture_somme_symbole_tore}
a_{\hbar}(x, y) = \sum_{(m, n) \in \mathbb{Z}^2} a_{m, n}^{\hbar} e^{ixn} e^{-i 2 \pi y m} 
\end{equation}
where $\left(a_{m, n}^{\hbar}\right)_{(m, n) \in \mathbb{Z}^2}$ is a sequence of complex coefficients depending on the semi-classical parameter $\hbar$. Recall an other definition of the Weyl quantization of a symbol on the torus, linked to Equation \eqref{equation_ecriture_somme_symbole_tore}, found in the book of Monique Cobescure and Didier Robert \cite[Chapter 6]{MR2952171}. By convention, this definition uses the parameter $k$, which is the inverse of the semi-classical parameter $ \hbar$. Throughout this text, we will make the abuse of notation of using $a_k$ and $a_{ \hbar}$ for the same object where $\hbar = 1/k$.

\begin{defi}[Second definition of the Weyl quantization of the torus] \label{defi_quantif_weyl_tore_somme}
Let $a_k \in \mathcal{C}^{ \infty}_k( \mathbb{R}^2)$ be a function such that, for all $(x, y) \in \mathbb{R}^2$, we have:
$$ a_k(x+ 2 \pi, y) = a_k(x, y) = a_k(x, y+1).$$
Define the Weyl quantization of the symbol $a_k$, denoted by $ \Op^w_k(a_k)$, by the following formula:
$$ \Op^w_k(a_k) = \sum_{(m, n) \in \mathbb{Z}^2} a_{m, n}^k \hat{T} \left( \dfrac{2 \pi m}{k}, \dfrac{n}{k} \right),$$
where the sequence $\left(a_{m,n}^k \right)_{(m,n) \in \mathbb{Z}^2}$ is defined by Equation \eqref{equation_ecriture_somme_symbole_tore} and where $ \hat{T}(p, q)$ is the Weyl-Heisenberg translation operator by a vector $(p, q) \in \mathbb{R}^2$ defined, for $ \phi \in \mathcal{S}( \mathbb{R})$, by:
$$ \hat{T}(p, q) \phi(x) = e^{-iqpk/2} e^{ixqk} \phi(x-p).$$ 
\end{defi}

\begin{prop}[\cite{MR2952171}] \label{prop_qauntif_weyl_Lk_avec_somme}
Let $a_k \in \mathcal{C}^{ \infty}_k( \mathbb{R}^2)$ be a function such that, for all $(x, y) \in \mathbb{R}^2$, we have:
$$ a_k(x + 2 \pi, y) = a_k(x, y) = a_k(x, y+1).$$
Then, we have: $ \Op^w_k(a_k) : \mathcal{L}_k \longrightarrow \mathcal{L}_k$.
\end{prop}

\begin{rema}[\cite{MR2679813}]
Definition \ref{defi_quantif_weyl_tore_integrale} and Definition \ref{defi_quantif_weyl_tore_somme} coincides in the sense that, if $a_{ \hbar} = a_k \in \mathcal{C}^{ \infty}_{ \hbar}( \mathbb{R}^2)$ is a function such that, for all $(x, y) \in \mathbb{R}^2$, we have:
$$ a_{\hbar}(x + 2 \pi, y) = a_{\hbar}(x, y) = a_{\hbar}(x, y+1).$$
Then, $ \Op^w(a_{ \hbar}) = \Op^w_k(a_k)$ on the space $ \mathcal{L}_k$.
\end{rema}

\subsubsection{Bargmann transform}
In this paragraph, we recall the definition of the semi-classical Bargmann transform and we study some of its properties. The principal difference with the transform introduced by Valentine Bargmann in the article \cite{MR0201959} is the weight function that we choose. The semi-classical Bargmann transform has been studied by Anders Melin, Michael Hitrik and Johannes Sjöstrand in \cite{MR1957486,MR2003421,MR2036816} and by the last two authors in the mini-course \cite{HJ}. Here, we investigate the action of the semi-classical Bargmann transform on the Schwartz space, on the tempered distributions space and on the space $ \mathcal{L}_k$. \\

First, we recall the definition of the Bargmann transform and its first properties (see for example the book of Maciej Zworski \cite[Chapter 13]{MR2952218}). 

\begin{defi}[Bargmann transform and its canonical transformation] \label{defi_Bargmann_transform}
Let $ \phi_1$ be the holomorphic quadratic function defined, for $(z, x) \in \mathbb{C} \times \mathbb{C}$, by:
$$ \phi_1(z,x) = \dfrac{i}{2} (z-x)^2 .$$
The Bargmann transform associated with the function $ \phi_1$ is the operator, denoted by $T_{ \phi_1}$, defined on $ \mathcal{S}( \mathbb{R})$ by:
$$ T_{ \phi_1} u(z) = c_{ \phi_1} \hbar^{-3/4} \int_{ \mathbb{R}} e^{(i/\hbar) \phi_1(z,x)} u(x) dx =  c_{ \phi_1} \hbar^{-3/4} \int_{ \mathbb{R}} e^{-(1/2 \hbar)(z - x)^2} u(x) dx,$$
where:
\begin{equation} \label{formule_constante_c_phi}
c_{ \phi_1} = \dfrac{1}{2^{1/2} \pi^{3/4}} \dfrac{| \det \partial_x \partial_z \phi_1|}{(\det \Im \partial^2_x \phi_1 )^{1/4}} = \dfrac{1}{2^{1/2} \pi^{3/4}}.
\end{equation}
Define the canonical transformation associated with $T_{ \phi_1}$ by:
\begin{align*}
\kappa_{ \phi_1}: \mathbb{C} \times \mathbb{C} & \longrightarrow \mathbb{C} \times \mathbb{C}, \\
(x, - \partial_x \phi_1(z, x)) =: (x, \xi) & \longmapsto (z, \partial_z \phi_1(z, x)) = (x- i\xi, \xi).
\end{align*}
\end{defi}

We have the following properties on the Bargmann transform (see for example \cite[Chapter 13]{MR2952218}).

\begin{prop} \label{prop_ecriture_integrale_Pi_Phi1} $ $
\begin{enumerate}
\item $ T_{ \phi_1}$ extends to a unitary transformation: $L^2( \mathbb{R}) \longrightarrow H_{ \hbar}( \mathbb{C}, \Phi_1)$.
\item If $ T_{ \phi_1}^* :L^2_{ \hbar}( \mathbb{C}, \Phi_1) \longrightarrow L^2( \mathbb{R})$ denotes the adjoint of $T_{ \phi_1}: L^2( \mathbb{R}) \longrightarrow L^2_{ \hbar}( \mathbb{C}, \Phi_1)$, then it is given by the following formula, for $v \in L^2_{ \hbar}( \mathbb{C}, \Phi_1)$:
$$ T_{ \phi_1}^* v(x) = c_{ \phi_1} \hbar^{-3/4} \int_{ \mathbb{C}} e^{-(1/2 \hbar) ( \overline{z} - x)^2} e^{-2 \Phi_1(z)/ \hbar} v(z) L( dz).$$
\item Let $ \psi_1$ be the unique holomorphic quadratic form on $ \mathbb{C} \times \mathbb{C}$ such that, for all $z \in \mathbb{C}$, we have:
$$ \psi_1(z, \overline{z}) = \Phi_1(z).$$
Then the orthogonal projection $ \Pi_{ \Phi_1, \hbar} : L^2_{ \hbar}( \mathbb{C}, \Phi_1) \longrightarrow H_{ \hbar}( \mathbb{C}, \Phi_1)$ is given by the following formula:
$$ \Pi_{ \Phi_1, \hbar} u(z) = \dfrac{2 \det \partial^2_{z, w} \psi_1}{\pi \hbar} \int_{ \mathbb{C}} e^{2( \psi_1(z, \overline{w})- \Phi_1(w))/ \hbar} u(w) dw d \overline{w}.$$
Moreover, $ \Pi_{ \Phi_1, \hbar} = T_{ \phi_1} T_{ \phi_1}^*$.
\end{enumerate}
\end{prop}

The following proposition gives a connection between the Weyl quantization of $ \mathbb{R}^2$ and the complex Weyl quantization of $ \mathbb{R}^2$ (see for example the mini-course \cite{HJ}).

\begin{prop}\label{prop_transfor_Bargmann_et_qauntif_weyl}
Let $a_{ \hbar} \in S( \mathbb{R}^2)$ be a function admitting an asymptotic expansion in powers of $ \hbar$. Let $ \Op^w_{ \Phi_1}(b_{ \hbar}) := T_{ \phi_1} \Op^w(a_{ \hbar}) T_{ \phi_1}^*$. Then:
\begin{enumerate}
\item[1.] $ \Op^w_{ \Phi_1}(b_{ \hbar}) : H_{ \hbar}( \mathbb{C}, \Phi_1) \longrightarrow H_{ \hbar}( \mathbb{C}, \Phi_1)$ is uniformly bounded with respect to $\hbar$;
\item[2.] $ \Op^w_{ \Phi_1}(b_{ \hbar})$ is given by the following contour integral:
$$ \Op^w_{ \Phi_1}(b_{ \hbar}) u(z) = \dfrac{1}{2 \pi \hbar} \int \! \! \! \int_{\Gamma(z)} e^{(i/\hbar) (z-w)\zeta} b_{ \hbar} \left( \dfrac{z+w}{2}, \zeta \right) u(w) dw d\zeta, $$
where $ \Gamma(z) = \left\lbrace (w, \zeta) \in \mathbb{C}^2; \zeta = \dfrac{2}{i} \dfrac{\partial \Phi_1}{\partial z} \left( \dfrac{z+w}{2} \right) = - \Im \left( \dfrac{z+w}{2} \right) \right\rbrace$, where the symbol $b_{ \hbar}$ is given by $ b_{ \hbar} = a_{ \hbar} \circ \kappa_{ \phi_1}^{-1}$ and where the canonical transformation $ \kappa_{ \phi_1}$ is defined by:
\begin{align*}
\kappa_{ \phi_1}: \mathbb{R}^2 & \longrightarrow \Lambda_{ \Phi_1} = \lbrace (z, - \Im(z)); z \in \mathbb{C} \rbrace \\
(x, \xi) & \longmapsto (x-i \xi, \xi).
\end{align*}
\end{enumerate} 
\end{prop}

We study now the action of the Bargmann transform on the Schwartz space in the spirit of the article of Valentine Bargmann \cite{MR0201959}, except that in our case, we introduce a semi-classical parameter and a different weight function. Therefore, for the sake of completeness, we recall the theory. To do so, we introduce some new notations. \\

\noindent \texttt{Notation:}
\begin{itemize}
\item for $ j \in \mathbb{N}$:
$$ \mathcal{S}^j( \mathbb{R}) := \left\lbrace \phi \in \mathcal{C}^j( \mathbb{R}); \| \phi \|_j := \max_{m \leq j} \left( \sup_{x \in \mathbb{R}} | (1+x^2)^{(j-m)/2} \partial^m_x \phi(x) | \right) < + \infty \right\rbrace ;$$
thus, the Schwartz space can be rewritten as follows:
$$ \mathcal{S}( \mathbb{R}) = \bigcap_{j=0}^{ \infty} \mathcal{S}^j( \mathbb{R}) = \left\lbrace \phi \in \mathcal{C}^{ \infty}( \mathbb{R}); \forall j \in \mathbb{N}, \| \phi \|_j < + \infty \right\rbrace ;$$
\item for $j \in \mathbb{N}$:
$$ \mathfrak{S}^j( \mathbb{C}) := \left\lbrace \psi \in \Hol( \mathbb{C}); | \psi|_j := \sup_{z \in \mathbb{C}} \left( \left(1+ |z|^2 \right)^{j/2} e^{- \Phi_1(z)/ \hbar} | \psi(z)| \right) < + \infty \right\rbrace;$$
\item we finally define:
$$ \mathfrak{S}( \mathbb{C}) := \bigcap_{j=0}^{ \infty} \mathfrak{S}^j( \mathbb{C}) = \left\lbrace \psi \in \Hol( \mathbb{C}); \forall j \in \mathbb{N}, | \psi |_j < + \infty \right\rbrace.$$
\end{itemize}

\begin{prop} \label{prop_Tphi_envoie_S(R)_sur_SS(C)} $ $
\begin{enumerate}
\item Let $j \in \mathbb{N}$, let $ \phi \in \mathcal{S}^j( \mathbb{R})$, then, for all $z \in \mathbb{C}$, we have the following estimate: 
\begin{equation} \label{equation_estimation_Tphi1_sur_Sj(R)}
| T_{ \phi_1} \phi(z) | \leq a_j^{ \hbar} \left( 1 + |z|^2 \right)^{-j/2} e^{ \Phi_1(z)/ \hbar} \| \phi \|_j,
\end{equation}
where $a_j^{ \hbar}$ is a constant depending on $j$ and on the semi-classical parameter $ \hbar$. As a result: $T_{ \phi_1} \mathcal{S}^j( \mathbb{R}) \subset \mathfrak{S}^j( \mathbb{C})$.
\item $T_{ \phi_1} \mathcal{S}( \mathbb{R}) \subset \mathfrak{S}( \mathbb{C})$.
\end{enumerate}
\end{prop}

\begin{proof}
We give a sketch of the proof ; for more details, see the article of Valentine Bargmann  \cite{MR0201959} where the argument can be adapted to the new weight. \\
\textbf{Step 1:} we prove using simple integral estimates that for $j=0$ and for $ \phi \in \mathcal{S}^0( \mathbb{R})$, there exists a constant $a_0^{ \hbar}$ such that, for all $z \in \mathbb{C}$, we have:
$$ | T_{ \phi_1} \phi(z) | \leq a_0^{ \hbar} e^{ \Phi_1(z)/ \hbar} \| \phi \|_0 .$$
\textbf{Step 2:} we prove that for $j \geq 1$ and for $ \phi \in \mathcal{S}^j( \mathbb{R})$, there exists a constant $a_j^{ \hbar}$ such that, for all $z \in \mathbb{C}$, we have:
$$ | T_{ \phi_1} \phi(z) | \leq a_j^{ \hbar} \left( 1 + |z|^2 \right)^{-j/2} e^{ \Phi_1(z)/ \hbar} \| \phi \|_j.$$
This step can be divided into 7 steps.
\begin{itemize}
\item \textbf{Step 2.1:} it follows from the definition of $\| \phi \|_j$ that:
\begin{enumerate}
\item[(a)] $ | \partial^m_x \phi(x) | \leq \| \phi \|_j$ for $m \leq j$;
\item[(b)] $ | \phi(x)| \leq \| \phi \|_j (1 + x^2)^{-j/2}$.
\end{enumerate}
\item \textbf{Step 2.2:} for $z \in \mathbb{C}$ and for $ \tau \in \mathbb{R}$, let: $F( \tau) = (T_{ \phi_1} \phi) ( \tau z)$, thus: $F(1) = (T_{ \phi_1} \phi)(z)$ and we can use the function $F$ to decompose $(T_{ \phi_1} \phi) (z)$ into two functions:
$$ F(1) = p_j(z) + r_j(z) ,$$
where:
$$ 
\left\lbrace
\begin{split}
p_j(z) & = \sum_{l=0}^{j-1} \dfrac{F^{(l)}(0)}{l!}, \\
r_j(z) & = \int_0^1 \dfrac{(1- \tau)^{j-1}}{(j-1)!} F^{(j)}( \tau) d \tau.
\end{split}
\right.$$
Then, we deduce the following estimates:
$$ 
\left\lbrace
\begin{split}
 \left|p_j(z) \right| & \leq \sum_{l=0}^{j-1} |z|^l \eta_l(0), \\
 \left| r_j(z) \right| & \leq \int_0^1 \dfrac{(1- \tau)^{j-1}}{(j-1)!} |z|^j \eta_j( \tau z) d \tau,
\end{split}
\right.$$
where $ \eta_l(z)$ is a bound on $ \partial^l (T_{ \phi_1} \phi) (z)$. \\
\item \textbf{Step 2.3:} we prove that the function $ \eta_l$ satisfies the following equality for $z \in \mathbb{C}$:
$$ \eta_l(z) = \beta e^{ \Phi_1(z)/ \hbar} \quad \text{for $l \leq j$},$$
where $ \beta = (\pi \hbar)^{-1/4} \| \phi \|_j$ using Lebesgue's theorem, Step 2.1 (a) and integral estimates. \\
\item \textbf{Step 2.4:} we deduce from Step 2.2 and Step 2.3 that, for $z \in \mathbb{C}$, we have the following estimates:
$$
\left\lbrace
\begin{split}
|p_j(z)| & \leq \beta \sum_{l=0}^{j-1} |z|^l, \\
|r_j(z)| & \leq \beta |z|^j (2 \hbar)^j e^{1/2 \hbar} (1+ \Im(z)^2)^{-j} e^{\Phi_1(z)/ \hbar}.
\end{split}
\right.$$
\item \textbf{Step 2.5:} we prove using Step 2.4 that, for $z \in \mathbb{C}$, we have:
$$
| T_{ \phi_1} \phi (z) | \leq \rho'_{ \hbar} \| \phi \|_j (1+|z|^2)^{j/2} (1+\Im(z)^2)^{-j} e^{ \Phi_1(z)/ \hbar} \quad \text{where $\rho'_{ \hbar}$ is a constant.}$$
\item \textbf{Step 2.6:} we prove using Step 2.1 (b) that, for $z \in \mathbb{C}$, we have:
$$
| T_{ \phi_1} \phi(z) |  \leq \rho_{ \hbar}'' \| \phi \|_j (1+\Re(z)^2)^{-j/2} e^{\Phi_1(z)/ \hbar} \quad \text{where $\rho_{ \hbar}''$ is an other constant.}$$
\item \textbf{Step 2.7:} we compare the estimates of Step 2.5 and Step 2.6 and we deduce that, for $z \in \mathbb{C}$, we have:
$$
| T_{ \phi_1} \phi(z) | \leq \rho_{ \hbar} \| \phi \|_j (1+|z|^2)^{-j/2} e^{\Phi_1(z)/ \hbar} \quad \text{where $ \rho_{ \hbar} = \max (2^j \rho_{ \hbar}', 2^{j/2} \rho_{ \hbar}'')$.}  $$
\end{itemize}
\textbf{Step 3:} the fact that $T_{ \phi_1} \mathcal{S}^j( \mathbb{R}) \subset \mathfrak{S}^j( \mathbb{C})$ is a corollary of Equation \eqref{equation_estimation_Tphi1_sur_Sj(R)} and $T_{ \phi_1} \mathcal{S}( \mathbb{R}) \subset \mathfrak{S}( \mathbb{C})$ can be deduced from the first assertion of the proposition and the definitions of the spaces $ \mathcal{S}( \mathbb{R})$ and $ \mathfrak{S}(\mathbb{C})$.
\end{proof}

\begin{rema} \label{rema_S(C)_inclus_dans_H(C,Phi1)}
Since $ \mathcal{S}( \mathbb{R}) \subset L^2( \mathbb{R})$, then according to Propositions  \ref{prop_ecriture_integrale_Pi_Phi1} and \ref{prop_Tphi_envoie_S(R)_sur_SS(C)}, we have $ \mathfrak{S}( \mathbb{C}) \subset H_{ \hbar}( \mathbb{C}, \Phi_1)$.
\end{rema}

Conversely, we have the following proposition.

\begin{prop} \label{prop_T_(phi1)*_envoie_S(C)_sur_S(R)} $ $ 
\begin{enumerate}
\item Let $ \mu = 1+j+ \tau$ with $j \in \mathbb{N}$ and $ \tau \in \mathbb{N}^*$, then, for all $ \psi \in \mathfrak{S}^{ \mu}( \mathbb{C})$, we have the following estimate:
\begin{equation} \label{equation_estimation_Tphi1*(SSmu(C))}
 \| T_{ \phi_1}^* \psi \|_j \leq a^{ \hbar}_{j, \tau} | \psi |_{ \mu},
\end{equation}
where $a^{ \hbar}_{j, \tau}$ is a constant depending on the semi-classical parameter $ \hbar$ and on the integers $j$ and $ \tau$. As a result: $ T_{ \phi_1}^* \mathfrak{S}^{ \mu}( \mathbb{C}) \subset \mathcal{S}^j( \mathbb{R})$.
\item $T_{ \phi_1}^* \mathfrak{S}( \mathbb{C}) \subset \mathcal{S}( \mathbb{R})$.
\end{enumerate} 
\end{prop}

\begin{proof}
We give a sketch of the proof, for more details see \cite{MR0201959}. \\
\textbf{Step 1:} we give an estimate on $ \partial^m T_{ \phi_1}^* \psi$ for $m \leq j$ by following these steps.
\begin{itemize}
\item \textbf{Step 1.1:} we prove using Lebesgue's theorem that, for $x \in \mathbb{R}$, we have:
$$ \left| \partial^m_x (T_{ \phi_1}^* \psi(x)) \right| \leq |c_{ \phi_1}| \hbar^{-3/4} \int_{ \mathbb{C}} B_m(z, x) L(dz), $$ 
where for $(z, x) \in \mathbb{C} \times \mathbb{R}$:
$$B_m(z, x) = \left| \partial^m_x \left( e^{-(1/2 \hbar)( \overline{z}-x)^2} \right) e^{-2 \Phi_1(z)/ \hbar} \psi(z) \right|.$$
\item \textbf{Step 1.2:} we prove that, for $(z, x) \in \mathbb{C} \times \mathbb{R}$ and for $m \leq j$, we have:
$$ B_m(z, x) \leq \delta_{\hbar}^j 2^m  \left( 1 + \dfrac{1}{2 \hbar}(\Re(z)-x)^2 \right)^{m/2} e^{-(1 /2 \hbar)(\Re(z)-x)^2} (1+|z|^2)^{(m- \mu)/2} | \psi |_{ \mu}, $$
where $ \delta_{\hbar}^j$ is a constant depending on  $j$ and $ \hbar$. To do so, we use estimates on Hermite polynomials for the term $ \left| \partial^m_x \left( e^{-(1/2 \hbar)( \overline{z}-x)^2} \right) \right| $ and the following estimate,  for $z \in \mathbb{C}$: 
$$| \psi(z) | \leq (1+|z|^2)^{- \mu/2} e^{ \Phi_1(z)/ \hbar} | \psi |_{ \mu} ,$$
resulting from the fact that $ \psi \in \mathfrak{S}^{ \mu}( \mathbb{C})$.
\item \textbf{Step 1.3:} according to Step 1.1 and Step 1.2, for $x \in \mathbb{R}$, we have:
\begin{align*}
| \partial_x^m (T_{\phi_1}^* \psi(x))|
& \leq a_{j, \tau}^{ \hbar} (1+x^2)^{(m-j)/2} | \psi |_{ \mu}, 
\end{align*}
where $a_{j, \tau}^{ \hbar}$ is a constant depending on $j$, $ \tau$ and $ \hbar$. 
\end{itemize}
\textbf{Step 2:} according to Step 1, for $x \in \mathbb{R}$ and $ \mu = 1+j + \tau$, we have:
\begin{align*}
(1+x^2)^{(j-m)/2} | \partial_x^m(T_{\phi_1}^* \psi(x))| \leq a_{j, \tau}^{ \hbar} | \psi |_{ \mu} .
\end{align*}
Thus:
$$ \sup_{x \in \mathbb{R}} \left( (1+x^2)^{(j-m)/2} | \partial_x^m(T_{\phi_1}^* \psi(x))| \right) \leq a_{j, \tau}^{ \hbar} | \psi |_{ \mu}.$$
Consequently, we have:
$$ \| T_{\phi_1}^* \psi \|_j = \max_{m \leq j} \left( \sup_{x \in \mathbb{R}} \left( (1+x^2)^{(j-m)/2} | \partial_x^m(T_{\phi_1}^* \psi(x))| \right) \right) \leq a_{j, \tau}^{ \hbar} | \psi |_{ \mu}.$$
\textbf{Step 3:} the fact that $T_{ \phi_1}^* \mathfrak{S}^{ \mu}( \mathbb{C}) \subset \mathcal{S}^j( \mathbb{R})$ is a corollary of Equation \eqref{equation_estimation_Tphi1*(SSmu(C))} and $T_{ \phi_1}^* \mathfrak{S}( \mathbb{C}) \subset \mathcal{S}( \mathbb{R})$ can be deduced from the first assertion of the proposition and the definitions of the spaces $ \mathcal{S}( \mathbb{R})$ and $ \mathfrak{S}(\mathbb{C})$.
\end{proof}

We are interested now in the action of the Bargmann transform on the tempered distributions space. Here again we can adapt the techniques of \cite{MR0201959}.\\

\noindent \texttt{Notation:}
\begin{itemize}
\item $ \mathfrak{S}'( \mathbb{C})$ denotes the dual of the space $ \mathfrak{S}( \mathbb{C})$ (equipped with the topology of the semi-norms $|\cdot |_j$) \textit{i.e.} the space of continuous linear functionals on $ \mathfrak{S}( \mathbb{C})$;
\item $ \langle \cdot, \cdot \rangle_{ \mathfrak{S}', \mathfrak{S}}$ denotes the duality bracket between $ \mathfrak{S}'( \mathbb{C})$ and $ \mathfrak{S}( \mathbb{C})$;
\item for $f, g \in \Hol( \mathbb{C})$, we denote by $ \langle g, f \rangle$ the following product:
$$ \langle g, f \rangle = \int_{ \mathbb{C}} \overline{g(z)} f(z) e^{- 2 \Phi_1(z)/ \hbar} L(dz),$$
when this integral converges.
\end{itemize}
 
\begin{rema} \label{rem_crochet_defini_sur_SSj(C)} $ $
\begin{itemize}
\item The bracket defined above coincides with the scalar product $ \langle g, f \rangle_{L^2_{ \hbar}( \mathbb{C}, \Phi_1)}$ when $g, f \in H_{ \hbar}( \mathbb{C}, \Phi_1)$. 
\item If $g \in \mathfrak{S}^{ \rho}( \mathbb{C})$ and $f \in \mathfrak{S}^{ \sigma}( \mathbb{C})$ with $ \rho + \sigma > 2$, then the bracket $ \langle g, f \rangle$ is well-defined.
\end{itemize}
\end{rema}

We use this bracket to describe the elements of the space $ \mathfrak{S}'( \mathbb{C})$ similarly to the article of Valentine Bargmann \cite{MR0201959}.

\begin{prop} \label{prop_ecriture_elements_de_SS'(C)}
Every continuous linear functional $L$ on $ \mathfrak{S}( \mathbb{C})$ can be written, for all $f \in \mathfrak{S}( \mathbb{C})$, as follows:
$$ L(f) = \langle g, f \rangle = \int_{ \mathbb{C}} \overline{g(z)} f(z) e^{-2 \Phi_1(z)/ \hbar} L(dz),$$
where $g$ is a function in $ \mathfrak{S}^{-l}( \mathbb{C})$ for $l \in \mathbb{N}$ and is uniquely defined, for all $a \in \mathbb{C}$, by $g(a) = \overline{L(e_a)}$ (where for all $z \in \mathbb{C}$, $e_a(z) = e^{-( \overline{a}-z)^2/(4 \hbar)}$). \\
Conversely, every functional of the form $L(f) = \langle g, f \rangle$ with $g \in \mathfrak{S}^{-l}( \mathbb{C})$ defines a continuous linear functional on $ \mathfrak{S}( \mathbb{C})$.
\end{prop}

\begin{proof}
We give a sketch of the proof, for more details see \cite{MR0201959}. \\
\textbf{Step 1:} for all $L \in \mathfrak{S}'( \mathbb{C})$, there exists $C>0$ and $l \in \mathbb{N}$ such that: $|L(f)| \leq C |f|_l$, for all $f \in \mathfrak{S}( \mathbb{C})$. \\
\textbf{Step 2:} for $a \in \mathbb{C}$, let $g$ be the function defined by $g(a) = \overline{L(e_a)}$. We prove using Step 1 that $e_a \in \mathfrak{S}( \mathbb{C})$ and $g \in \mathfrak{S}^{-l} ( \mathbb{C})$ ($e_a$ is a reproducing kernel for the space $H_{ \hbar}( \mathbb{C}, \Phi_1)$ and for $ \mathfrak{S}( \mathbb{C})$). \\
\textbf{Step 3:} let $L_1$ be the continuous linear functional defined, for $f \in \mathfrak{S}( \mathbb{C})$, by $L_1(f) = \langle g, f \rangle$. We show that, for all $a \in \mathbb{C}$, we have $ L_1(e_a) = L(e_a)$ then we deduce that $L=L_1$ using the density of the set of finite linear combinations of elements of $ \mathcal{B} = \lbrace e_a, a \in \mathbb{C} \rbrace$ in $ \mathfrak{S}( \mathbb{C})$.
\end{proof}

As in the article of Valentine Bargmann \cite{MR0201959}, we prove that the Bargmann transform $T_{ \phi_1}$ and its adjoint $T_{ \phi_1}^*$ act on the spaces $ \mathcal{S}'( \mathbb{R})$ and $ \mathfrak{S}'( \mathbb{C})$ respectively.

\begin{prop} \label{prop_Tphi_envoie_S'(R)_sur_SS'(C)}$ $ 
\begin{enumerate}
\item $T_{ \phi_1}$ extends to an operator: $\mathcal{S}'( \mathbb{R}) \longrightarrow \mathfrak{S}'( \mathbb{C})$, which satisfies for $v \in \mathcal{S}'( \mathbb{R})$ and $f \in \mathfrak{S}( \mathbb{C})$:
$$ \langle T_{ \phi_1} v, f \rangle_{ \mathfrak{S}', \mathfrak{S}} = \langle v, T_{ \phi_1}^* f \rangle_{ \mathcal{S}', \mathcal{S}} .$$
\item $T_{ \phi_1}^*$ extends to an operator: $\mathfrak{S}'( \mathbb{C}) \longrightarrow \mathcal{S}'( \mathbb{R})$, which satisfies for $L \in \mathfrak{S}'( \mathbb{C})$ and $ \phi \in \mathcal{S}( \mathbb{R})$:
$$ \langle T_{ \phi_1}^* L, \phi \rangle_{ \mathcal{S}', \mathcal{S}} = \langle L, T_{ \phi_1} \phi \rangle_{ \mathfrak{S}', \mathfrak{S}}.$$
\end{enumerate}
\end{prop}

\begin{proof}
We give a sketch of the proof, for more details see \cite{MR0201959}. \\
Let $v \in \mathcal{S}'( \mathbb{R})$, let $L(f)$ be the functional defined by:
$$ L(f) = \langle v, \phi \rangle_{ \mathcal{S}', \mathcal{S}},$$
where $f = T_{ \phi_1} \phi \in \mathfrak{S}( \mathbb{C})$, then $L(f)$ is a continuous linear functional on $ \mathfrak{S}( \mathbb{C})$. Conversely if $L \in \mathfrak{S}'( \mathbb{C})$ and if we define $v$ by:
$$ v( \phi) = \langle v, \phi \rangle_{ \mathcal{S}', \mathcal{S}} = L(f) \quad \text{where $f =T_{ \phi_1} \phi$},$$
then $v$ is a continuous linear functional on $ \mathcal{S}( \mathbb{R})$. Then, according to Proposition \ref{prop_ecriture_elements_de_SS'(C)}, for $l \in \mathbb{N}$, there exists $g \in \mathfrak{S}^{-l}( \mathbb{C})$ such that:
$$ L(f) = \langle g, f \rangle .$$
Thus, for all $v \in \mathcal{S}'( \mathbb{R})$ and for all $ \phi \in \mathcal{S}( \mathbb{R})$, we have:
$$ \langle v, \phi \rangle_{ \mathcal{S}', \mathcal{S}} = \langle g, T_{ \phi_1} \phi \rangle .$$
This equality gives a bijection between the spaces $ \mathcal{S}'( \mathbb{R})$ and $ \mathfrak{S}'( \mathbb{C})$ with:
$$ g := T_{ \phi_1} v \quad \text{and} \quad v =: T_{ \phi_1}^* g .$$
\end{proof}

\begin{rema} \label{prop_reecriture_action_Tphi1_sur_S'(R)}
For $\psi \in \mathcal{S}'( \mathbb{R})$, we can rewrite $T_{ \phi_1} \psi$ as follows (see for example \cite{HJ} or \cite{MR2952218}):
$$ T_{ \phi_1} \psi(z) = \left\langle \psi, c_{ \phi_1} \hbar^{-3/4} e^{-(1/ 2 \hbar)(z-.)^2} \right\rangle_{ \mathcal{S}', \mathcal{S}} .$$ 
\end{rema}

We are now looking at the range by the Bargmann transform of the space $ \mathcal{L}_k$ and we prove that this range is the space $ \mathcal{H}_k$, where we recall that:
\begin{align*}
\mathcal{L}_k & = \left\lbrace \psi \in \mathcal{S}'( \mathbb{R}); \quad \tau_{2 \pi} \psi  = u^k \psi, \quad \tau_1 \mathcal{F}_{ \hbar}( \psi) = v^{-k} \mathcal{F}_{ \hbar}( \psi)  \right\rbrace, \\
\mathcal{H}_k & =  \left\lbrace g \in \Hol( \mathbb{C}); \quad g(p+2 \pi, q) = u^k g(p, q), \quad g(p, q+1) = v^k e^{-i(p+iq)k+k/2} g(p, q) \right\rbrace.
\end{align*}
To our knowledge, this result is new in the literature and it constitutes a fundamental step in our proof of Theorem \ref{theoA}.

\begin{prop} \label{prop_formule_C}
Let $k \geq 1$. Then, we have:
\begin{enumerate}
\item $T_{ \phi_1}: \mathcal{L}_k \longrightarrow \mathcal{H}_k $;
\item $T_{ \phi_1}^*: \mathcal{H}_k \longrightarrow \mathcal{L}_k $.
\end{enumerate}

\end{prop}

\begin{proof}
According to Proposition \ref{prop_Tphi_envoie_S'(R)_sur_SS'(C)}, $T_{ \phi_1}: \mathcal{S}'( \mathbb{R}) \longrightarrow \mathfrak{S}'( \mathbb{C})$. Since $ \mathcal{L}_k \subset \mathcal{S}'( \mathbb{R})$, then $T_{ \phi_1}$ is well-defined on this space. Let's prove that the Bargmann transform $T_{ \phi_1}$ sends the basis $( \epsilon_l)_{l \in \mathbb{Z}/ k \mathbb{Z}}$ of $ \mathcal{L}_k$ (see Equation \eqref{eq_base_epsilon_l}) on the basis $( e_l)_{l \in \mathbb{Z}/ k \mathbb{Z}}$ of $ \mathcal{H}_k$ (see Equation \eqref{eq_base_el}). Let $c$ be the real number such that $u = e^{ic}$. Let $l \in \lbrace 0, \ldots, k-1 \rbrace$, using Remark \ref{prop_reecriture_action_Tphi1_sur_S'(R)}, we have:
\begin{align*}
T_{ \phi_1} \epsilon_l(z) & = \left\langle \epsilon_l, c_{ \phi_1} \hbar^{-3/4} e^{-(1/ 2 \hbar)(z-.)^2} \right\rangle_{ \mathcal{S}', \mathcal{S}}, \\
& = \left\langle u^{k./(2 \pi)} \sum_{j \in \mathbb{Z}} \left(v^{-k}\right)^j e^{i(l+jk).}, c_{ \phi_1} \hbar^{-3/4} e^{-(1/ 2 \hbar)(z-.)^2}  \right\rangle_{ \mathcal{S}', \mathcal{S}}, \\
& = c_{ \phi_1} k^{3/4} \sum_{j \in \mathbb{Z}} \left(v^{-k} \right)^j \left\langle u^{k./(2 \pi)} e^{i(l+jk).},  e^{-(k/ 2)(z-.)^2} \right\rangle_{ \mathcal{S}', \mathcal{S}}  \quad \text{since $k = \dfrac{1}{\hbar}$}, \\
& = c_{ \phi_1} k^{3/4} u^{kz/(2 \pi)} \sum_{j \in \mathbb{Z}} \left(v^{-k} \right)^j e^{i(l+jk)z} \left\langle u^{k./(2 \pi)} e^{i(l+jk).} , e^{-(k/2) (.)^2} \right\rangle_{ \mathcal{S}', \mathcal{S}}, \\
& = c_{ \phi_1} k^{3/4} u^{kz/(2 \pi)} \sum_{j \in \mathbb{Z}} \left( v^{-k} \right)^j e^{i(l+jk)z} \sqrt{\dfrac{2 \pi}{k}} \exp \left(- \dfrac{1}{2k}  \left( \dfrac{ck}{2 \pi}+l+jk \right)^2\right).
\end{align*}
By a simple computation, we obtain:
$$ \dfrac{1}{2 k}\left( \dfrac{ck}{2 \pi}+l+jk \right)^2 = \dfrac{1}{2k} \left( \dfrac{ck}{2 \pi} +l \right)^2 + \dfrac{j^2k}{2} + jl + \dfrac{jck}{2 \pi} .$$
Therefore, we have:
\begin{align*}
T_{ \phi_1} \epsilon_l(z) & = c_{ \phi_1} k^{3/4} u^{kz/(2 \pi)} \sum_{j \in \mathbb{Z}} \left(v^{-k} \right)^j e^{i(l+jk)z} \sqrt{\dfrac{2 \pi}{k}} \exp \left(- \dfrac{ 1}{2k} \left( \dfrac{ck}{2 \pi}+l+jk \right)^2 \right), \\
& = c_k^l u^{kz/(2 \pi)} \sum_{j \in \mathbb{Z}} \left(v^{-k} e^{-jk/2-l} u^{ik/(2 \pi)} \right)^j e^{i(l+jk)z}, \\
& = c_k^l e_l(z),
\end{align*}
where $c_k^l$ is given by the following equality:
$$ c_k^l = c_{ \phi_1} \sqrt{2\pi} k^{1/4} \exp \left( - \dfrac{1}{2k} \left( \dfrac{ck}{2 \pi} +l \right)^2 \right) .$$
Conversely, we compute $T_{ \phi_1}^* e_l$. First, since $ \mathcal{H}_k \subset \mathfrak{S}^0( \mathbb{C})$ and since the function $z \longmapsto e^{-(1/ 2 \hbar)(z-x)^2}$ belongs to the space $\mathfrak{S}( \mathbb{C})$, then for $v \in \mathcal{H}_k$, we have (according to Remark \ref{rem_crochet_defini_sur_SSj(C)}):
$$ \left\langle c_{ \phi_1} \hbar^{-3/4} e^{-(1/ 2 \hbar)(.-x)^2}, v \right\rangle =  c_{ \phi_1} \hbar^{-3/4} \int_{ \mathbb{C}} e^{-1/(2 \hbar)( \overline{z} - x)^2} e^{-2 \Phi_1(z)/ \hbar} e_l(z) L(dz) < + \infty .$$
Let $l \in \lbrace 0, 1, \ldots, k-1 \rbrace$, then we have:
\begin{align*}
& T_{ \phi_1}^* e_l(x) \\
&  = c_{ \phi_1} \hbar^{-3/4} \int_{ \mathbb{C}} e^{-1/(2 \hbar)( \overline{z} - x)^2} e^{-2 \Phi_1(z)/ \hbar} e_l(z) L(dz), \\
& = c_{ \phi_1} k^{3/4} \int_{ \mathbb{C}} e^{-(k/2)( \overline{z} -x)^2} e^{-k ( \Im z)^2} e_l(z) L(dz) \quad \text{because $ \Phi_1(z) = \dfrac{1}{2} ( \Im z)^2$ and $k = \dfrac{1}{\hbar}$ }, \\
& =c_{ \phi_1} k^{3/4} \int_{ \mathbb{C}} e^{-(k/2)( \overline{z} -x)^2} e^{-k ( \Im z)^2} u^{kz/(2 \pi)} \sum_{j \in \mathbb{Z}} \left(v^{-k} e^{-l-jk/2} u^{ik/(2 \pi)} \right)^j e^{i(l+jk)z} L(dz), \\
& = c_{ \phi_1} k^{3/4} \sum_{j \in \mathbb{Z}} \left(v^{-k} e^{-l-jk/2} u^{ik/(2 \pi)} \right)^j \int_{ \mathbb{C}} e^{-(k/2)(\overline{z})^2} e^{-k ( \Im (z+x))^2} u^{k(z+x)/(2 \pi)} e^{i(l+jk)(z+x)} L(dz), \\
& = c_{ \phi_1} k^{3/4} u^{kx/(2 \pi)} \sum_{j \in \mathbb{Z}} \left(v^{-k} e^{-l-jk/2} u^{ik/(2 \pi)} \right)^j e^{i(l+jk)x} \int_{ \mathbb{C}} e^{-(k/2)(\overline{z})^2} e^{-k ( \Im z)^2} e^{iz(l+jk+ck/(2 \pi))} L(dz).
\end{align*}
We have to compute the following integral (after the change of variables $z=p+iq$):
\begin{align*}
& \int_{ \mathbb{R}} \int_{ \mathbb{R}} e^{-(k/2)(p-iq)^2} e^{-kq^2} e^{i(p+iq)(l+jk+ck/(2 \pi))} dp dq \\
& = \int_{ \mathbb{R}} \int_{ \mathbb{R}} e^{-k p^2/2} e^{-kq^2/2} e^{ikpq} e^{ip(l+jk+ck/(2 \pi))} e^{-q(l+jk+ck/(2 \pi))} dp dq.
\end{align*}
By a simple computation, we obtain:
\begin{align*}
\int_{ \mathbb{R}} e^{-k p^2/2} e^{ikpq} e^{ip(l+jk+ck/(2 \pi))} dp = \sqrt{\dfrac{2\pi}{k}} \exp \left( - \dfrac{1}{2k} \left( l+jk+ \dfrac{ck}{2 \pi} \right)^2 - \dfrac{kq^2}{2} - q \left( l+jk+ \dfrac{ck}{2 \pi} \right) \right).
\end{align*}
Thus, by an other simple computation, we obtain:
\begin{align*}
\int_{ \mathbb{R}} \int_{ \mathbb{R}} e^{-kp^2/2} e^{-kq^2/2} e^{ikpq} e^{ip(l+jk+ck/(2 \pi))} e^{-q(l+jk+ck/(2 \pi))} dp dq  = \sqrt{2} \dfrac{\pi}{k} e^{(l+jk+ck/(2 \pi))^2/(2k)}.
\end{align*}
Consequently, we obtain:
\begin{align*}
& T_{ \phi_1}^* e_l(x) \\ 
& = c_{ \phi_1} k^{3/4} \sqrt{2} \dfrac{\pi}{k} u^{kx/(2 \pi)} \sum_{j \in \mathbb{Z}} \left(v^{-k} e^{-l-jk/2} u^{ik/(2 \pi)} \right)^j e^{i(l+jk)x}  e^{(l+jk+ck/(2 \pi))^2/(2k)}, \\
& = c_{ \phi_1} k^{-1/4} \sqrt{2} \pi e^{(l+ck/(2 \pi))^2/(2k)} u^{kx/(2 \pi)} \sum_{j \in \mathbb{Z}} \left(v^{-k} \right)^j e^{i(l+jk)x}, \\
& = \tilde{c}_k^l \epsilon_l(x),
\end{align*}
where $\tilde{c}_k^l = c_{ \phi_1} k^{-1/4} \sqrt{2} \pi e^{(l+ck/(2 \pi))^2/(2k)}$.
\end{proof}

Since the Bargmann transform is a unitary transformation between the spaces $L^2( \mathbb{R})$ and $H_{ \hbar}( \mathbb{C}, \Phi_1)$, we study this feature between the spaces $ \mathcal{L}_k$ and $ \mathcal{H}_k$.

\begin{prop} \label{prop_TphiTphi*=id_sur_Hk_et_Tphi*Tphi=id_sur_Lk}  $ $
\begin{enumerate}
\item $T_{ \phi_1}^* T_{ \phi_1} = \id$ on $\mathcal{L}_k$.
\item $T_{ \phi_1} T_{ \phi_1}^*= \id$ on $ \mathcal{H}_k$.
\end{enumerate}
\end{prop}

\begin{proof}
Let $c$ be the real number such that $u = e^{ic}$. According to the proof of Proposition \ref{prop_formule_C}, we have, for $l \in \lbrace 0, \ldots, k-1 \rbrace$:
$$ T_{ \phi_1} \epsilon_l = c_k^l e_l \quad \text{with $c_k^l = c_{ \phi_1} \sqrt{2\pi} k^{1/4} e^{-( ck/(2 \pi) +l)^2/(2k)}$}.$$
Let $ \mathcal{C} = \diag (c_k^0, \ldots, c_k^{k-1})$ be the matrix of the operator $T_{\phi_1}$ in the basis $(e_l)_{l \in \mathbb{Z}/ k \mathbb{Z}}$. According to the proof of Proposition \ref{prop_formule_C}, we also have, for $l \in \lbrace 0, \ldots, k-1 \rbrace$:
$$ T_{\phi_1}^* e_l = \tilde{c}_k^l \epsilon_l \quad \text{with $\tilde{c}_k^l = c_{ \phi_1} k^{-1/4} \sqrt{2} \pi e^{(l+ck/(2 \pi))^2/(2k)}$}.$$
Let $ \mathcal{C}^* = \diag( \tilde{c}_k^0, \ldots, \tilde{c}_k^{k-1}) $ be the matrix of the operator $T_{ \phi_1}^*$ in the basis $(\epsilon_l)_{l \in \mathbb{Z}/ k \mathbb{Z}}$. We want to prove that: $ \mathcal{C} \mathcal{C}^* = \mathcal{C}^* \mathcal{C} = I_k$. Let $k \geq 1$ and let $l \in \lbrace 0, 1, \ldots, k-1 \rbrace$, we have:
\begin{align*}
c_k^l \tilde{c}_k^l & = c_{ \phi_1} \sqrt{2\pi} k^{1/4} e^{-( ck/(2 \pi) +l)^2/(2k)} c_{ \phi_1} k^{-1/4} \sqrt{2} \pi e^{(l+ck/(2 \pi))^2/(2k)}, \\
& = c_{ \phi_1}^2 2 \pi^{3/2}, \\
& = \left( \dfrac{1}{2^{1/2} \pi^{3/4}} \right)^2 2 \pi^{3/2} \quad \text{according to Definition \ref{defi_Bargmann_transform}}, \\
& = 1, \\
& = \tilde{c}_k^l c_k^l.
\end{align*}
Therefore, we have: $ \mathcal{C} \mathcal{C}^* = \mathcal{C}^* \mathcal{C} = I_k $.
\end{proof}

\subsubsection{Complex Weyl quantization of the torus}
In this paragraph, we define the complex Weyl quantization of a symbol on the torus. As in the  classical Weyl quantization case, we have two definitions for the complex Weyl quantization. With an Egorov theorem analogous to Proposition \ref{prop_transfor_Bargmann_et_qauntif_weyl} in the $ \mathbb{R}^2$-case, we exhibit the notion of complex Weyl quantization of the torus. First, we introduce a new class of symbols. Recall that $ \Lambda_{ \Phi_1}$ denotes the following space:
$$ \Lambda_{ \Phi_1} = \left\lbrace \left( z, \dfrac{2}{i} \dfrac{\partial \Phi_1}{\partial z}(z) \right); z \in \mathbb{C} \right\rbrace = \left\lbrace (z, - \Im(z)); z \in \mathbb{C} \right\rbrace.$$
And that the canonical transformation $ \kappa_{ \phi_1}$ is defined as follows:
\begin{align*}
\kappa_{ \phi_1}: \mathbb{R}^2 & \longrightarrow \Lambda_{ \Phi_1} \\
(x, y) & \longmapsto (z, w) := (x-iy, y).
\end{align*}	
Notice that, if $a_{k} \in \mathcal{C}^{ \infty}_k( \mathbb{R}^2)$ is a function such that, for all $(x, y) \in \mathbb{R}^2$, we have:
$$ a_{k}(x + 2 \pi, y) = a_{k}(x, y) = a_{k}(x, y+1) ;$$
and if $b_{k}$ is the function defined by the following relation, for $(z, w) \in \Lambda_{ \Phi_1}$:
$$ b_{k} (z, w) := a_{k} \circ \kappa_{ \phi_1}^{-1}(z, w).$$
Then $b_{k} \in \mathcal{C}^{ \infty}_k( \Lambda_{ \Phi_1})$ is a function such that, for $(z, w) \in \Lambda_{ \Phi_1}$, we have:
$$ b_{k}(z+2 \pi, w) = b_{k}(z, w) = b_{k}(z+i, w-1).$$
Besides, thanks to the identification of $ \Lambda_{ \Phi_1}$ with $ \mathbb{C}$, we can rewrite the symbol $b_{k}$ as a convergent series, for $z \in \mathbb{C} \simeq \Lambda_{ \Phi_1}$:
\begin{equation} \label{equation_ecriture_somme_symbole_sur_LambdaPhi1}
b_{k}(z) = \sum_{(m, n) \in \mathbb{Z}^2} b_{m,n}^{ k} e^{in \Re(z)} e^{2 i \pi m \Im(z)},
\end{equation}
where $\left(b_{m,n}^{k} \right)_{(m,n) \in \mathbb{Z}^2}$ is defined by the following formula:
$$ b_{m,n}^{k} = a_{m,n}^{k} \quad \text{where for $(x, y) \in \mathbb{R}^2$,} \quad a_{k}(x, y) = \sum_{(m,n) \in \mathbb{Z}^2} a_{m,n}^{k} e^{inx} e^{-2 i \pi my} .$$

\begin{rema}
Since $ \Lambda_{ \Phi_1} \simeq \mathbb{C}$, then the class of symbols $S( \Lambda_{ \Phi_1})$ can be identified with $S( \mathbb{C}) \simeq S( \mathbb{R}^2)$.
\end{rema}

We can now deduce the following Egorov theorem.

\begin{prop} \label{prop_quantif_complexe_somme_definie_sur_S(C)} $ $ \\
Let $a_{k} \in \mathcal{C}^{ \infty}_k( \mathbb{R}^2)$ be a function such that, for all $(x, y) \in \mathbb{R}^2$, we have:
$$ a_{k}(x+2 \pi, y) = a_{k}(x, y) = a_{k}(x, y+1).$$
Then, we have:
$$ T_{ \phi_1} \Op^w_k(a_{k}) = \Op^w_{ \Phi_1, k}(a_{k} \circ \kappa_{ \phi_1}^{-1}) T_{ \phi_1} \quad \text{on $ \mathcal{S}( \mathbb{R})$},$$
where $ \Op^w_{ \Phi_1, k}$ is defined by the following formula, for $u \in \mathfrak{S}( \mathbb{C})$:
$$ \Op_{\Phi_1, k}^w(a_{k} \circ \kappa_{ \phi_1}^{-1}) u(z) = \sum_{(m, n) \in \mathbb{Z}^2} a_{m,n}^k e^{-i \pi mn/k} e^{-n^2/2k} e^{inz} u\left( z- \dfrac{2 \pi m}{k} + \dfrac{in}{k} \right) ,$$
where $\left(a_{m,n}^k \right)_{(m, n) \in \mathbb{Z}^2}$ is the sequence of coefficients defined in Equation \eqref{equation_ecriture_somme_symbole_tore}.
\end{prop}

\begin{proof}
According to Definition \ref{defi_quantif_weyl_tore_somme}, $ \Op^w_k(a_k) : \mathcal{S}( \mathbb{R}) \longrightarrow \mathcal{S}( \mathbb{R})$. Let $\phi \in \mathcal{S}( \mathbb{R})$, then we have:
\begin{align*}
& T_{ \phi_1} ( \Op^w_k(a_k) \phi)(z) \\
& = c_{ \phi_1}  \hbar^{-3/4} \int_{ \mathbb{R}} e^{-(1/2 \hbar)(z-x)^2} (\Op^w_k (a_k) \phi)(x) dx, \\
& = c_{ \phi_1}  \hbar^{-3/4} \int_{ \mathbb{R}} e^{-(1/2 \hbar)(z-x)^2} \sum_{(m, n) \in \mathbb{Z}^2} a_{m,n}^k e^{-i \pi mn/k} e^{ixn} \phi \left( x - \dfrac{2 \pi m}{k} \right) dx, \\
& = c_{ \phi_1}  \hbar^{-3/4} \int_{ \mathbb{R}} \sum_{(m, n) \in \mathbb{Z}^2} a_{m,n}^k e^{-i \pi mn/k} e^{izn} e^{-(1/2 \hbar)(z-x)^2} e^{-i(z-x)n} \phi \left( x - \dfrac{2 \pi m}{k} \right) dx, \\
& = c_{ \phi_1}  \hbar^{-3/4} \int_{ \mathbb{R}} \sum_{(m, n) \in \mathbb{Z}^2} a_{m,n}^k e^{-i \pi mn/k} e^{izn} e^{-(1/2 \hbar)(z-x+in \hbar)^2} e^{-n^2 \hbar/2} \phi \left( x - \dfrac{2 \pi m}{k} \right) dx, \\
& = c_{ \phi_1}  \hbar^{-3/4} \int_{ \mathbb{R}} \sum_{(m, n) \in \mathbb{Z}^2} a_{m,n}^k e^{-i \pi mn/k} e^{izn} e^{-(1/2 \hbar)(z-x+in \hbar-2 \pi m \hbar)^2} e^{-n^2 \hbar/2} \phi \left( x \right) dx, \\
& =\sum_{(m, n) \in \mathbb{Z}^2} a_{m,n}^k e^{-i \pi mn/k} e^{-n^2 \hbar/2} e^{izn} c_{ \phi_1}  \hbar^{-3/4} \int_{ \mathbb{R}}  e^{-(1/2 \hbar)(z-x+in \hbar-2 \pi m \hbar)^2} \phi \left( x \right) dx, \\
& = \sum_{(m, n) \in \mathbb{Z}^2} a_{m,n}^k e^{-i \pi mn/k} e^{-n^2/2k} e^{izn} (T_{ \phi_1} \phi) \left( z - \dfrac{2 \pi m}{k} + \dfrac{in}{k} \right) \quad \text{because $ \hbar = \dfrac{1}{k}$}.
\end{align*}
Therefore, for $u \in \mathfrak{S}( \mathbb{C})$, we define $ \Op_{ \Phi_1,k}^w(a_k \circ \kappa_{ \phi_1}^{-1})$ as follows:
$$ \Op_{ \Phi_1, k}^w(a_k \circ \kappa_{ \phi_1}^{-1}) u(z) = \sum_{(m, n) \in \mathbb{Z}^2} a_{m,n}^k e^{-i \pi mn/k} e^{-n^2/2k} e^{izn} u \left( z - \dfrac{2 \pi m}{k} + \dfrac{in}{k} \right).$$
\end{proof}

The previous proposition leads us to define the notion of complex Weyl quantization of a symbol on the torus as follows. 

\begin{defi}[First definition of the complex Weyl quantization of the torus] \label{defi_weyl_complex_tore_1}
Let $b_k \in \mathcal{C}^{ \infty}_k( \Lambda_{ \Phi_1})$ be a function such that, for all $(z, w) \in \Lambda_{ \Phi_1}$, we have:
$$ b_k(z+2 \pi, w) = b_k(z, w) = b_k(z+i, w-1).$$
Define the complex Weyl quantization of the symbol $b_k$, denoted by $ \Op^w_{ \Phi_1,k}(b_k)$, by the following formula, for $u \in \mathfrak{S}( \mathbb{C})$:
$$ \Op_{\Phi_1, k}^w(b_k) u(z) = \sum_{(m, n) \in \mathbb{Z}^2} b_{m,n}^k e^{-i \pi mn/k} e^{-n^2/2k} e^{inz} u\left( z- \dfrac{2 \pi m}{k} + \dfrac{in}{k} \right) ,$$
where the sequence $\left(b_{m,n}^k \right)_{(m,n) \in \mathbb{Z}^2}$ is given by Equation \eqref{equation_ecriture_somme_symbole_sur_LambdaPhi1}.
\end{defi}

Let's prove a basic property on this notion of quantization. The operator $ \Op^w_{\Phi_1,k}(b_k)$ defined above acts on the space $ \mathfrak{S}( \mathbb{C})$. We are now going to show that it also acts on the space $ \mathcal{H}_k$ (as expected since the Bargmann transform sends the space $ \mathcal{L}_k$ on the space $ \mathcal{H}_k$).

\begin{prop} \label{prop_Op_Weyl_k_agit_sur_Hk}
Let $b_k \in \mathcal{C}^{ \infty}_k( \Lambda_{ \Phi_1})$ be a function such that, for all $(z, w) \in \Lambda_{ \Phi_1}$, we have:
$$ b_k(z+ 2 \pi, w) = b_k(z, w) = b_k(z+i, w-1) .$$ 
Then $ \Op^w_{\Phi_1, k}(b_k)$ can be extended into an operator which sends the space $\mathcal{H}_k$ on itself.
\end{prop}

\begin{proof}
According to Proposition \ref{prop_quantif_complexe_somme_definie_sur_S(C)}, $ \Op^w_{\Phi_1, k}(b_k) : \mathfrak{S}( \mathbb{C}) \longrightarrow \mathfrak{S}( \mathbb{C})$. Let $u, v \in \mathfrak{S}( \mathbb{C})$, then we have:
\begin{align*}
& \langle \Op^w_{\Phi_1, k}(b_k) u, v \rangle_{ \mathfrak{S}, \mathfrak{S}} \\
& = \int_{ \mathbb{C}} \overline{\Op^w_{\Phi_1, k}(b_k) u(z)} v(z) e^{-2 \Phi_1(z)/ \hbar} L(dz), \\
& = \int_{ \mathbb{C}} \sum_{(m, n) \in \mathbb{Z}^2} \overline{b_{m,n}^k} e^{i \pi mn/k} e^{-n^2/2k} e^{-in\overline{z}} \overline{u\left( z- \dfrac{2 \pi m}{k} + \dfrac{in}{k} \right) } v(z) e^{-2 \Phi_1(z)/ \hbar} L(dz), \\
& = \int_{ \mathbb{C}} \sum_{(m, n) \in \mathbb{Z}^2} \overline{b_{m,n}^k} e^{-i \pi mn/k} e^{-n^2/2k} e^{-in\overline{z}} \overline{u\left( z + \dfrac{in}{k} \right) } v \left( z + \dfrac{2 \pi m}{k} \right) e^{-2 \Phi_1(z)/ \hbar} L(dz), \\
& = \int_{ \mathbb{C}} \overline{u\left( z \right) } \sum_{(m, n) \in \mathbb{Z}^2} \overline{b_{m,n}^k} e^{-i \pi mn/k} e^{-n^2/2k} e^{-inz}  v \left( z + \dfrac{2 \pi m}{k} - \dfrac{in}{k} \right) e^{-2 \Phi_1(z)/ \hbar} L(dz), \\
& = \int_{ \mathbb{C}} \overline{u\left( z \right) } \sum_{(m, n) \in \mathbb{Z}^2} \overline{b_{-m,n}^k} e^{i \pi mn/k} e^{-n^2/2k} e^{-inz}  v \left( z - \dfrac{2 \pi m}{k} - \dfrac{in}{k} \right) e^{-2 \Phi_1(z)/ \hbar} L(dz) \quad \text{via $m \longmapsto -m$}, \\
& = \int_{ \mathbb{C}} \overline{u\left( z \right) } \sum_{(m, n) \in \mathbb{Z}^2} \overline{b_{-m,-n}^k} e^{-i \pi mn/k} e^{-n^2/2k} e^{inz}  v \left( z - \dfrac{2 \pi m}{k} + \dfrac{in}{k} \right) e^{-2 \Phi_1(z)/ \hbar} L(dz) \quad \text{via $n \longmapsto -n$}, \\
& = \int_{ \mathbb{C}} \overline{u\left( z \right) } \Op^w_{ \Phi_1,k}( \overline{b}_k) v(z) e^{-2 \Phi_1(z)/ \hbar} L(dz), \\
& = \langle u, \Op^w_{ \Phi_1,k}( \overline{b}_k) v \rangle_{ \mathfrak{S}, \mathfrak{S}},
\end{align*}
where $ \overline{b}_k \in \mathcal{C}^{ \infty}_k( \Lambda_{ \Phi_1})$ is defined, for $(z, w) \in \Lambda_{ \Phi_1}$, by:
\begin{align*}
\overline{b}_k(z, w) & = \sum_{(m, n)\in \mathbb{Z}^2} \overline{b_{-m, -n}^k} e^{in(z+iw)} e^{-2i \pi m w}, \\
& = \sum_{(m, n)\in \mathbb{Z}^2} \overline{b_{m, n}^k} e^{-in(z+iw)} e^{2i \pi m w} \quad \text{via $(m,n) \longmapsto (-m,-n)$}, \\
& = \sum_{(m, n)\in \mathbb{Z}^2} \overline{b_{m, n}^k} e^{-in\Re(z)} e^{-2i \pi m \Im(z)} \quad \text{because $(z, w) \in \Lambda_{ \Phi_1}$, thus $w = - \Im(z)$},\\
& = \overline{b_k(z, w)}.
\end{align*}
Since $v \in \mathfrak{S}( \mathbb{C})$ and $ \overline{b}_k \in \mathcal{C}^{ \infty}_k( \Lambda_{ \Phi_1})$, then $ \Op^w_{ \Phi_1,k}( \overline{b}_k) v \in \mathfrak{S}( \mathbb{C})$ and the complex Weyl quantization $ \Op^w_{ \Phi_1, k}(b_k)$ is well-defined on $ \mathfrak{S}'( \mathbb{C})$ by the following formula, for $u \in \mathfrak{S}'( \mathbb{C})$ and for $v \in \mathfrak{S}( \mathbb{C})$:
$$ \langle \Op^w_{ \Phi_1,k}(b_k) u, v \rangle_{ \mathfrak{S}', \mathfrak{S}} : = \langle u, \Op^w_{ \Phi_1, k}( \overline{b}_k) v \rangle_{ \mathfrak{S}', \mathfrak{S}}.$$
Afterwards, since $ \mathcal{H}_k \subset \mathfrak{S}'( \mathbb{C})$, then for $g \in \mathcal{H}_k$, $ \Op^w_{ \Phi_1,k}(b_k) g \in \Hol( \mathbb{C})$. Moreover, for $g \in \mathcal{H}_k$, using simple computations we prove that:
$$
\left\lbrace
\begin{split}
& \Op^w_{ \Phi_1,k}(b_k) g(z+ 2 \pi) = u^k \Op^w_{ \Phi_1, k}(b_k) g(z), \\
& \Op^w_{ \Phi_1,k}(b_k)  g(z+i) = v^k e^{-ikz+k/2}  \Op^w_{ \Phi_1, k}(b_k) g(z).
\end{split}
\right.$$
\end{proof}

Let's give a second definition of the complex Weyl quantization of the torus. This notion is analogous to the already existing one in the $ \mathbb{R}^2$-case (see for example \cite{HJ} or \cite{MR2952218}). We believe that it is the first time that such a contour integral is used in a context of the quantization of a compact phase space.

\begin{defi}[Second definition of the complex Weyl quantization of the torus] \label{defi_weyl_complex_tore_2}
Let $b_{ \hbar} \in \mathcal{C}^{ \infty}_{ \hbar}( \Lambda_{ \Phi_1})$ be a function such that, for all $(z, w) \in \Lambda_{ \Phi_1}$, we have:
$$ b_{\hbar}(z+2 \pi, w) = b_{\hbar}(z, w) = b_{\hbar}(z+i, w-1).$$
Define the complex Weyl quantization of the symbol $b_{ \hbar}$, denoted by $ \Op^w_{ \Phi_1}(b_{\hbar})$, by the following formula, for $u \in \mathfrak{S}( \mathbb{C})$:
$$ \Op_{ \Phi_1}^w(b_{ \hbar}) u(z) = \dfrac{1}{2 \pi \hbar} \int \! \! \! \int_{ \Gamma(z)} e^{(i/ \hbar)(z-w) \zeta} b_{ \hbar} \left( \dfrac{z+w}{2}, \zeta \right) u(w) dw d \zeta,$$
where the contour integral is the following:
$$ \Gamma(z) = \left\lbrace (w, \zeta) \in \mathbb{C}^2; \zeta =  \dfrac{2}{i} \dfrac{\partial \Phi_1}{\partial z} \left( \dfrac{z+w}{2} \right) = - \Im \left( \dfrac{z+w}{2} \right) \right\rbrace.$$
\end{defi}

This second definition expresses the fact that the complex Weyl quantization of $ \mathbb{R}^2$ (seen in Proposition \ref{prop_transfor_Bargmann_et_qauntif_weyl}) can be extended to a symbol defined on the torus. Similarly to Proposition \ref{prop_Op_Weyl_k_agit_sur_Hk}, we have the following property.

\begin{prop} \label{prop_action_quantif_weyl_complexe_sur_SS'(C)}
Let $b_{ \hbar} \in \mathcal{C}^{ \infty}_{ \hbar}( \Lambda_{ \Phi_1})$ be a function such that, for all $(z, w) \in \Lambda_{ \Phi_1}$, we have: 
$$b_{ \hbar}(z + 2 \pi, w) = b_{ \hbar}(z, w) = b_{ \hbar}(z+i, w-1).$$
Then, $ \Op^w_{ \Phi_1}(b_{ \hbar})$ can be extended into an operator which sends $\mathcal{H}_k$ on itself.
\end{prop}

\begin{proof}
Let $a_{ \hbar}:= b_{ \hbar} \circ \kappa_{ \phi_1}$, then $a_{ \hbar} \in  S( \mathbb{R}^2)$ and $ \Op^w(a_{ \hbar}): \mathcal{S}( \mathbb{R}) \longrightarrow \mathcal{S}( \mathbb{R})$.
According to Proposition \ref{prop_Tphi_envoie_S(R)_sur_SS(C)}, $ T_{ \phi_1} :\mathcal{S}( \mathbb{R}) \longrightarrow \mathfrak{S}(  \mathbb{C})$ and according to Proposition \ref{prop_transfor_Bargmann_et_qauntif_weyl}, we have:
$$ \Op^w_{ \Phi_1}(b_{ \hbar}) = T_{ \phi_1} \Op^w(a_{ \hbar}) T_{ \phi_1}^* : \mathfrak{S}( \mathbb{C}) \longrightarrow \mathfrak{S}( \mathbb{C}) .$$
Afterwards, let $u, v \in \mathfrak{S}( \mathbb{C})$, then we have:
\begin{align*}
& \langle \Op^w_{ \Phi_1}(b_{ \hbar}) u, v \rangle_{ \mathfrak{S}, \mathfrak{S}} \\
& = \int_{ \mathbb{C}} \overline{(\Op^w_{ \Phi_1} (b_{ \hbar}) u)(z)} v(z) e^{-2 \Phi_1(z)/ \hbar} L(dz), \\
& = \int_{ \mathbb{C}} \dfrac{1}{2 \pi \hbar} \int \! \! \! \int_{ \Gamma(z)} e^{(-i/ \hbar) \overline{(z-w) \zeta}} \overline{ b_{ \hbar} \left( \dfrac{z+w}{2}, \zeta \right)} \overline{ u(w)} dw d \zeta v(z) e^{-2 \Phi_1(z)/ \hbar} L(dz), \\
& = \int_{ \mathbb{C}} \dfrac{C}{2 \pi \hbar} \int_{ \mathbb{C}} e^{(-i/ \hbar) \overline{(z-w)(- \Im(z+w/2))}} \overline{ b_{ \hbar} \left( \dfrac{z+w}{2}, - \Im \left( \dfrac{z+w}{2} \right) \right)} \overline{ u(w)} L(dw) v(z) e^{-2 \Phi_1(z)/ \hbar} L(dz),
\end{align*}
where we used the definition of the contour integral $ \Gamma(z)$ and where $C>0$ is a constant. Then, for $z \in \mathbb{C} \simeq \Lambda_{ \Phi_1}$, we have:
$$
b_{ \hbar}(z, - \Im(z)) = \sum_{(m, n) \in \mathbb{Z}^2} b_{m,n}^{ \hbar} e^{in \Re(z)} e^{2i \pi m \Im(z)}.$$
Thus, for $z \in \Lambda_{ \Phi_1}$, we obtain:
$$
\overline{b_{ \hbar}(z, - \Im(z))} = \sum_{(m, n) \in \mathbb{Z}^2} \overline{b_{m,n}^{ \hbar}} e^{-in \Re(z)} e^{-2i \pi m \Im(z)} = \overline{b}_{ \hbar}(z, - \Im(z)) .$$
Therefore, we can rewrite the integral as follows, for $u, v \in \mathfrak{S}( \mathbb{C})$:
\begin{align*}
& \langle \Op^w_{ \Phi_1}(b_{ \hbar}) u, v \rangle_{ \mathfrak{S}, \mathfrak{S}} \\
& = \dfrac{C}{2 \pi \hbar} \int_{ \mathbb{C}} \int_{ \mathbb{C}} e^{(i/ \hbar) ( \overline{z}- \overline{w})\Im(z+w/2)} \overline{b}_{ \hbar} \left( \dfrac{z+w}{2}, - \Im \left( \dfrac{z+w}{2} \right) \right) \overline{ u(w)} v(z) e^{-2 \Phi_1(z)/ \hbar}  L(dw) L(dz).
\end{align*}
With a short computation, we prove the following equality:
\begin{align*}
\dfrac{i}{\hbar} \left( \overline{z} - \overline{w} \right) \Im \left( \dfrac{z+w}{2} \right) - \dfrac{2}{\hbar} \Phi_1(z)
& = \dfrac{i}{h}(z-w) \Im \left( \dfrac{z+w}{2} \right) - \dfrac{2}{\hbar} \Phi_1(w).
\end{align*}
Consequently, for all $u, v \in \mathfrak{S}( \mathbb{C})$, we have:
\begin{align*}
& \langle \Op^w_{ \Phi_1}(b_{ \hbar}) u, v \rangle_{ \mathfrak{S}, \mathfrak{S}} \\
& = \int_{ \mathbb{C}} \overline{u(w)} \dfrac{C}{2 \pi \hbar} \int_{ \mathbb{C}} e^{(i/ \hbar) (w-z)(-\Im(z+w/2))} \overline{b}_{ \hbar} \left( \dfrac{z+w}{2}, - \Im \left( \dfrac{z+w}{2} \right) \right) v(z) L(dz) e^{-2 \Phi_1(w)/ \hbar}  L(dw) , \\
& = \int_{ \mathbb{C}} \overline{u(w)} \dfrac{1}{2 \pi \hbar} \int \! \! \! \int_{ \Gamma(w)} e^{(i/ \hbar) (w-z)\zeta} \overline{b}_{ \hbar} \left( \dfrac{z+w}{2}, \zeta \right) v(z) L(dz) e^{-2 \Phi_1(w)/ \hbar}  L(dw) , \\
& = \int_{ \mathbb{C}} \overline{u(w)} (\Op^w_{ \Phi_1}( \overline{b}_{ \hbar})v)(w) e^{-2 \Phi_1(w)/ \hbar} L(dw), \\
& = \langle u, \Op^w_{ \Phi_1}( \overline{b}_{ \hbar}) v \rangle_{ \mathfrak{S}, \mathfrak{S}},
\end{align*}
where $ \Gamma(w) = \left\lbrace (z, \zeta) \in \mathbb{C}^2; \zeta = - \Im \left( \dfrac{z+w}{2} \right) \right\rbrace$. Since, for $v \in \mathfrak{S}( \mathbb{C})$, $ \Op^w_{ \Phi_1}( \overline{b}_{ \hbar}) v \in \mathfrak{S}( \mathbb{C})$, then the operator $ \Op^w_{ \Phi_1}(b_{ \hbar})$ is well-defined on $\mathfrak{S}'( \mathbb{C})$ by the following equality, for $u \in \mathfrak{S}'( \mathbb{C})$ and $v \in \mathfrak{S}( \mathbb{C})$:
$$ \langle \Op^w_{ \Phi_1}(b_{ \hbar}) u, v \rangle_{ \mathfrak{S}', \mathfrak{S}} = \langle u, \Op^w_{ \Phi_1}( \overline{b}_{ \hbar}) v \rangle_{ \mathfrak{S}', \mathfrak{S}}.$$
As a result, the operator $ \Op^w_{ \Phi_1}(b_{\hbar})$ is well-defined on $ \mathcal{H}_k$ because it is a subspace of $ \mathfrak{S}'( \mathbb{C})$. Afterwards, since for $ g \in \mathcal{H}_k$, $ \Op^w_{ \Phi_1}(b_{ \hbar}) g \in \mathfrak{S}'( \mathbb{C})$, then we have $ \Op^w_{ \Phi_1}(b_{ \hbar}) g \in \Hol( \mathbb{C})$. Besides using simple computations, for $g \in \mathcal{H}_k$, we prove that:
$$
\left\lbrace
\begin{split}
& \Op^w_{ \Phi_1}(b_{\hbar}) g(z+ 2 \pi)  = u^k \Op^w_{ \Phi_1}(b_{ \hbar}) g(z), \\
& \Op^w_{ \Phi_1}(b_{ \hbar}) g(z+i) = v^k e^{-ikz+k/2}  \Op^w_{ \Phi_1}(b_{ \hbar}) g(z).
\end{split}
\right.$$
\end{proof}

To conclude this paragraph, we link Definition \ref{defi_weyl_complex_tore_1} and Definition \ref{defi_weyl_complex_tore_2}.

\begin{prop} Let $b_{ \hbar} =b_k \in \mathcal{C}^{ \infty}_{ \hbar}( \Lambda_{ \Phi_1})$ be a function such that, for all $(z, w) \in \Lambda_{ \Phi_1}$, we have:
$$b_{ \hbar}(z+ 2 \pi, w) = b_{ \hbar}(z, w) = b_{ \hbar}(z+i, w-1).$$ 
Then:
$$ \Op^w_{ \Phi_1}(b_{ \hbar}) = \Op^w_{\Phi_1,k}(b_k) \quad \text{on $ \mathcal{H}_k$.} $$
\end{prop}

\begin{proof}
According to Equation \eqref{equation_ecriture_somme_symbole_sur_LambdaPhi1}, for $(z, w) \in \Lambda_{ \Phi_1}$, we can rewrite $b_{ \hbar}$ as follows:
$$ b_{ \hbar}(z, w) = \sum_{(m, n) \in \mathbb{Z}^2} b_{m,n}^{ \hbar} e^{in(z+iw)} e^{-2 i \pi m w}.$$
Consequently, we obtain, for $u \in \mathfrak{S}( \mathbb{C})$:
\begin{align*}
& \Op_{ \Phi_1}^w(b_{ \hbar}) u(z) \\
& =\dfrac{1}{2 \pi \hbar} \int \! \! \!\int_{ \Gamma(z)} e^{(i/ \hbar)(z-w) \zeta} \sum_{(m,n) \in \mathbb{Z}^2} b_{m,n}^{ \hbar} e^{in((z+w)/2+i \zeta)} e^{-2i \pi m \zeta} u(w) dw d \zeta, \\
& = \dfrac{1}{2 \pi \hbar} \int \! \! \!\int_{ \Gamma(z)} \sum_{(m,n) \in \mathbb{Z}^2} e^{(i/ \hbar)(z-w) \zeta}  b_{m,n}^{ \hbar} e^{in((z+w)/2 +i \zeta)}  e^{-i\pi m n/k} u \left( w - \dfrac{2 \pi m}{k} \right) dw d \zeta, \\
& =\dfrac{1}{2 \pi \hbar} \int \! \! \!\int_{ \Gamma(z)} e^{(i/ \hbar)(z-w) \zeta} \sum_{(m,n) \in \mathbb{Z}^2}   b_{m,n}^{ \hbar} e^{-i\pi m n/k} e^{-n^2/2k} u \left( w + \dfrac{in}{k} - \dfrac{2 \pi m}{k} \right) dw d \zeta, \\
& = \dfrac{1}{2 \pi \hbar} \int \! \! \!\int_{ \Gamma(z)} e^{(i/ \hbar)(z-w) \zeta} (\Op^w_{\Phi_1,k} (b_k) u)(w) dw d \zeta, \\
& = \Op^w_{\Phi_1,k} (b_k) u(z),
\end{align*}
where we used the change of variables:
$ \Gamma(z) \ni (w, \zeta) \longmapsto  \left( w+ \dfrac{in}{k}, \zeta - \dfrac{n}{2k} \right) \in \Gamma(z)$. Therefore, for $u \in \mathfrak{S}'( \mathbb{C})$ and $v \in \mathfrak{S}( \mathbb{C})$, we have:
\begin{align*}
\langle \Op^w_{ \Phi_1}(b_{ \hbar}) u, v \rangle_{ \mathfrak{S}', \mathfrak{S}} & = \langle u, \Op^w_{ \Phi_1}( \overline{b}_{ \hbar}) v \rangle_{ \mathfrak{S}', \mathfrak{S}}, \\
& = \langle u, \Op^w_{ \Phi_1,k}( \overline{b}_k) v \rangle_{ \mathfrak{S}', \mathfrak{S}}, \\
& = \langle \Op^w_{ \Phi_1,k}(b_k) u, v \rangle_{ \mathfrak{S}', \mathfrak{S}}.
\end{align*}
In others words, we have:
$$ \Op^w_{ \Phi_1}(b_{ \hbar}) = \Op^w_{ \Phi_1,k}(b_k) \quad \text{on $ \mathfrak{S}'( \mathbb{C})$}.$$
Finally, since $ \mathcal{H}_k \subset \mathfrak{S}'( \mathbb{C})$, then by restriction and according to Propositions \ref{prop_Op_Weyl_k_agit_sur_Hk} and \ref{prop_action_quantif_weyl_complexe_sur_SS'(C)}, we obtain the result.
\end{proof}

\subsection{Connections between the quantizations of the torus $ \mathbb{T}^2$} \label{subsection_links_between_quantizations}
In this paragraph, we relate the different notions of quantization of the torus. To do so, we follow these steps:
\begin{enumerate}
\item[1.] we recall the connection between a Berezin-Toeplitz operator and a complex pseudo-differential operator of the complex plane;
\item[2.] within the Berezin-Toeplitz setting, we relate the quantization of the torus to the quantization of the complex plane;
\item[3.] we establish a correspondence between the Berezin-Toeplitz and the complex Weyl quantizations of the torus.
\end{enumerate}

\subsubsection{Berezin-Toeplitz and complex Weyl quantizations of the complex plane}

First, we recall the definition of the Berezin-Toeplitz quantization of a symbol on the complex plane and then the definition of the complex Weyl quantization of a symbol on $ \Lambda_{ \Phi_1}$ (see for example \cite{MR2952218}). 

\begin{defi}[Berezin-Toeplitz quantization of the complex plane] \label{defi_quantif_BT_du_plan_complexe}
Let $f_k \in S( \mathbb{C})$ be a function admitting an asymptotic expansion in powers of $1/k$. Define the Berezin-Toeplitz quantization of $f_k$ by the sequence of operators $T_{f_k} := (T_k)_{ k \geq 1}$, where for $k \geq 1$, $T_k$ is defined by:
$$ T_k = \Pi_{ \Phi_1,k} M_{f_k} \Pi_{ \Phi_1,k} ,$$
where $M_{f_k}: L^2_k( \mathbb{C}, \Phi_1) \longrightarrow  L^2_k( \mathbb{C}, \Phi_1)$ is the multiplication operator by the function $f_k$ and where we recall that $ \Pi_{ \Phi_1, k}$ is the orthogonal projection of the space $L^2_k( \mathbb{C}, \Phi_1)$ on $H_k( \mathbb{C}, \Phi_1)$ defined in Proposition \ref{prop_ecriture_integrale_Pi_Phi1}. \\
We call $f_k$ the symbol of the Berezin-Toeplitz operator $T_{f_k}$. 
\end{defi}

\begin{defi}[Complex Weyl quantization of the complex plane]
Let $b_{ \hbar} \in S( \Lambda_{ \Phi_1})$ be a function admitting an asymptotic expansion in powers of $ \hbar$. Define the complex Weyl quantization of $b_{ \hbar}$, denoted by $ \Op^w_{ \Phi_1}(b_{ \hbar})$, by the following formula, for $u \in H_{ \hbar}( \mathbb{C}, \Phi_1)$:
$$ \Op_{ \Phi_1}^w(b_{ \hbar}) u(z) = \dfrac{1}{2 \pi \hbar} \int \! \! \! \int_{ \Gamma(z)} e^{(i/ \hbar)(z-w) \zeta} b_{ \hbar} \left( \dfrac{z+w}{2}, \zeta \right) u(w) dw d \zeta,$$
where the contour integral is the following:
$$ \Gamma(z) = \left\lbrace (w, \zeta) \in \mathbb{C}^2; \zeta =  \dfrac{2}{i} \dfrac{\partial \Phi_1}{\partial z} \left( \dfrac{z+w}{2} \right) = - \Im \left( \dfrac{z+w}{2} \right) \right\rbrace.$$
\end{defi}

Recall now the result relating these two quantizations (see for example \cite[Chapter 13]{MR2952218}).

\begin{prop} \label{prop_Toeplitz=pseudo_H(Phi1)}  $ $
\begin{enumerate}
\item Let $f_k \in S( \mathbb{C})$ be a function admitting an asymptotic expansion in powers of $1/k$. Let $T_{f_k} = (T_k)_{k \geq 1}$ be the Berezin-Toeplitz operator of symbol $f_k$. Then, for $k \geq 1$, we have: 
$$ T_k = \Op^w_{ \Phi_1}( b_{ \hbar}) \quad \text{on $H_k( \mathbb{C}, \Phi_1)$},$$
where $b_{ \hbar} \in S( \Lambda_{ \Phi_1})$ is a function admitting an asymptotic expansion in powers of $ \hbar$ given by the following formula, for all $z \in \Lambda_{ \Phi_1} \simeq \mathbb{C}$:
$$ b_{ \hbar}(z) = \exp \left( \dfrac{1}{k} \partial_z \partial_{ \overline{z}} \right) ( f_k(z)).$$
\item Let $b_{ \hbar} \in S( \Lambda_{ \Phi_1})$ be a function admitting an asymptotic expansion in powers of $ \hbar$. Then, there exists $f_k \in S( \mathbb{C})$ a function admitting an asymptotic expansion in powers of $1/k$ such that for $k \geq 1$:
$$ \Op^w_{ \Phi_1}(b_{ \hbar}) = T_k + \mathcal{O}(k^{- \infty}) \quad \text{on $H_{ \hbar}( \mathbb{C}, \Phi_1)$} ,$$
where $(T_k)_{k \geq 1} = T_{f_k}$ is the Berezin-Toeplitz operator of symbol $f_k$ and where, for all $N \in \mathbb{N}$ and for $z \in \mathbb{C}$, $f_k$ is given by:
$$ f_k(z) = \sum_{j=0}^N \dfrac{\hbar^j}{j!} \left( D_z D_{ \overline{z}} \right)^j (b_{ \hbar}(z)) + \mathcal{O}(\hbar^{N+1}) .$$
\end{enumerate}
\end{prop}

\subsubsection{Berezin-Toeplitz quantization of the torus and Berezin-Toeplitz quantization of the complex plane}
In this paragraph, we study a Berezin-Toeplitz operator of the complex plane whose symbol is $2 \pi$-periodic with respect to its first variable and $1$-periodic with respect to its second variable. Previously, we looked at the action of a Berezin-Toeplitz operator of the complex plane on the spaces $ \mathfrak{S}( \mathbb{C})$ and $ \mathfrak{S}'( \mathbb{C})$.

\begin{prop} \label{prop_toeplitz_sur_C_defini_sur_S'} Let $f_k \in S( \mathbb{C})$ be a function admitting an asymptotic expansion in powers of $1/k$. Let $T_{f_k} = (T_k)_{k \geq 1}$ be the Berezin-Toeplitz operator of symbol $f_k$. Then, for $k \geq 1$, we have:
\begin{enumerate}
\item $T_{k}$ can be defined as an operator which sends $\mathfrak{S}( \mathbb{C})$ on itself by:
$$ T_k v = \Pi_{ \Phi_1,k}(f_k v) \quad \text{for $v \in \mathfrak{S}( \mathbb{C})$,}$$
where $ \Pi_{\Phi_1,k}$ is seen as an operator which sends $ \mathfrak{S}( \mathbb{C})$ on itself.
\item $T_{k}$ can be extended into an operator which sends $\mathfrak{S}'( \mathbb{C})$ on itself by:
$$ \langle T_k u, v \rangle_{ \mathfrak{S}', \mathfrak{S}} = \langle u, \tilde{T}_k v \rangle_{ \mathfrak{S}', \mathfrak{S}} \quad \text{for $u \in \mathfrak{S}'( \mathbb{C})$ and for $v \in \mathfrak{S}( \mathbb{C})$,}$$
where $ (\tilde{T}_k)_{k \geq 1} =: T_{ \overline{f}_k}$ is the Berezin-Toeplitz operator of symbol $ \overline{f}_k$.
\end{enumerate}
\end{prop}

\begin{proof}
Since $T_{f_k} = (T_k)_{k \geq 1}$ is a Berezin-Toeplitz operator of the complex plane, then according to Proposition \ref{prop_Toeplitz=pseudo_H(Phi1)}, there exists $b_{ \hbar } \in S( \Lambda_{ \Phi_1})$ such that, for $k \geq 1$:
$$ T_k = \Op^w_{ \Phi_1}(b_{ \hbar}) \quad \text{on $H_k( \mathbb{C}, \Phi_1)$}.$$
Besides, according to Proposition \ref{prop_transfor_Bargmann_et_qauntif_weyl}, we know that:
$$ T_{ \phi_1}^* \Op^w_{ \Phi_1}(b_{ \hbar}) T_{ \phi_1} = \Op^w( b_{ \hbar} \circ \kappa_{ \phi_1}) : L^2( \mathbb{R}) \longrightarrow L^2( \mathbb{R}).$$
Since $b_{ \hbar} \circ \kappa_{ \phi_1} \in S( \mathbb{R}^2)$, then we have:
$$ \Op^w(b_{ \hbar} \circ \kappa_{ \phi_1}) : \mathcal{S}( \mathbb{R}) \longrightarrow \mathcal{S}( \mathbb{R}) \quad \text{and} \quad \Op^w(b_{ \hbar} \circ \kappa_{ \phi_1}) : \mathcal{S}'( \mathbb{R}) \longrightarrow \mathcal{S}'( \mathbb{R}) .$$
Moreover, according to Propositions \ref{prop_Tphi_envoie_S(R)_sur_SS(C)}, \ref{prop_T_(phi1)*_envoie_S(C)_sur_S(R)} and \ref{prop_Tphi_envoie_S'(R)_sur_SS'(C)}, the Bargmann transform and its adjoint satisfy:
$$
\left\lbrace
\begin{split}
& T_{ \phi_1}: \mathcal{S}( \mathbb{R}) \longrightarrow \mathfrak{S}( \mathbb{C}) \quad \text{and} \quad T_{ \phi_1}^*: \mathfrak{S}( \mathbb{C}) \longrightarrow \mathcal{S}( \mathbb{R}), \\
& T_{ \phi_1}: \mathcal{S}'( \mathbb{R}) \longrightarrow \mathfrak{S}'( \mathbb{C}) \quad \text{and} \quad T_{ \phi_1}^*: \mathfrak{S}'( \mathbb{C}) \longrightarrow \mathcal{S}'( \mathbb{R}).
\end{split}
\right.
$$
As a result, for $k \geq 1$, we obtain:
$$ T_k: \mathfrak{S}( \mathbb{C}) \longrightarrow \mathfrak{S}( \mathbb{C}) \quad \text{and} \quad T_k: \mathfrak{S}'( \mathbb{C}) \longrightarrow \mathfrak{S}'( \mathbb{C}) .$$
Then, by definition $ \Pi_{ \Phi_1,k} = T_{ \phi_1} T_{ \phi_1}^*$ and according to Proposition \ref{prop_Tphi_envoie_S(R)_sur_SS(C)}, the operator $ \Pi_{ \Phi_1, k}$ can be extended into an operator which sends $ \mathfrak{S}( \mathbb{C})$ on itself. Since $ \mathfrak{S}( \mathbb{C}) \subset H_{ \hbar}( \mathbb{C}, \Phi_1)$ (see Remark \ref{rema_S(C)_inclus_dans_H(C,Phi1)}), then for $v \in \mathfrak{S}( \mathbb{C})$ and for $f_k \in S( \mathbb{C})$, we have:
$$ \Pi_{ \Phi_1,k} v =v \quad \text{and} \quad \Pi_{ \Phi_1, k}(f_k v) \in \mathfrak{S}( \mathbb{C}).$$ Therefore,  the Berezin-Toeplitz operator $T_{f_k} = (T_k)_{ k \geq 1}$ is defined as follows, for $k \geq 1$ and for $v \in \mathfrak{S}( \mathbb{C})$:
$$ T_k v = \Pi_{ \Phi_1,k} M_{f_k} \Pi_{ \Phi_1,k} v = \Pi_{ \Phi_1,k}(f_k v),$$
where $ \Pi_{ \Phi_1,k}$ is an operator which sends $ \mathfrak{S}( \mathbb{C})$ on itself. Finally, for $v = T_{\phi_1} \psi \in \mathfrak{S}'( \mathbb{C})$ and for $ u = T_{ \phi_1} \phi \in \mathfrak{S}( \mathbb{C})$, we have:
\begin{align*}
\langle T_k v, u \rangle_{ \mathfrak{S}', \mathfrak{S}} & = \langle \Op^w_{ \Phi_1}(b_{ \hbar}) v, u \rangle_{ \mathfrak{S}', \mathfrak{S}} \quad \text{according to Proposition \ref{prop_Toeplitz=pseudo_H(Phi1)}}, \\
& = \langle v, \Op^w_{ \Phi_1}( \overline{b}_{ \hbar}) u \rangle_{ \mathfrak{S}', \mathfrak{S}} \quad \text{according to Proposition \ref{prop_action_quantif_weyl_complexe_sur_SS'(C)}}, \\
& =:  \langle v, \tilde{T}_k u \rangle_{ \mathfrak{S}', \mathfrak{S}}.
\end{align*}
Then, according to Proposition \ref{prop_Toeplitz=pseudo_H(Phi1)}, for $z \in \Lambda_{ \Phi_1} \simeq \mathbb{C}$, we have:
$$b_{ \hbar}(z) = \exp \left( \dfrac{1}{k} \partial_z \partial_{ \overline{z}} \right) ( f_k(z)).$$
Consequently, for $z \in \Lambda_{ \Phi_1} \simeq \mathbb{C}$, we obtain:
$$ \overline{b}_{ \hbar}(z) = \exp \left( \dfrac{1}{k} \partial_z \partial_{ \overline{z}} \right) (\overline{f_k(z)}).$$
In others words, the sequence of operators $( \tilde{T}_k)_{k \geq 1}$ is a Berezin-Toeplitz operator of symbol $ \overline{f}_k$.
\end{proof}

\begin{rema}
Let $f_k \in \mathcal{C}^{ \infty}_k( \mathbb{R}^2)$ be a function such that, for $ (x, y) \in \mathbb{R}^2$, we have:
$$ f_k(x+2 \pi, y) = f_k(x, y) = f_k(x, y+1).$$
Let $T_{f_k} = (T_k)_{k \geq 1}$ be the Berezin-Toeplitz operator of the complex plane of symbol $f_k$. Then, for $k \geq 1$, the operator $T_{k}$ is well-defined on $ \mathcal{H}_k$ according to Proposition \ref{prop_toeplitz_sur_C_defini_sur_S'} since $ \mathcal{H}_k \subset \mathfrak{S}'( \mathbb{C})$.
\end{rema}

The following proposition gives a connection between the orthogonal projection $ \Pi_{ \Phi_1,k}$ which appears in the definition of a Berezin-Toeplitz operator of the complex plane (see Definition \ref{defi_quantif_BT_du_plan_complexe}) and the orthogonal projection $ \Pi_k$ which appears in the definition of a Berezin-Toeplitz operator of the torus (see Definition \ref{defi_quantif_BT_tore}). This proposition is fundamental for understanding the relation between these two quantizations.

\begin{prop} \label{prop_Pi_Phi_bien_defini_sur_Hk} $ $ \\
Let $ \Pi_{ \Phi_1, k}$ be the orthogonal projection of $L^2_k( \mathbb{C}, \Phi_1)$ on $ H_k( \mathbb{C}, \Phi_1)$. Then:
\begin{enumerate}
\item $ \Pi_{ \Phi_1, k}$ can be extended into an operator which sends $\mathcal{G}_k$ on $\mathcal{H}_k$ (defined in Subsection \ref{subsection_context});
\item $ \Pi_{ \Phi_1, k}= \id$ on $ \mathcal{H}_k$. 
\end{enumerate}
Consequently, $ \Pi_{ \Phi_1, k}$ coincides with $ \Pi_k$ on $ \mathcal{G}_k$.
\end{prop}

\begin{proof} First, let's prove that $ \Pi_{ \Phi_1,k}$ is well-defined on $ \mathcal{G}_k$. The main difficulty to prove this result comes from the fact that $ \mathcal{G}_k$ is not included in $ \mathfrak{S}'( \mathbb{C})$. Recall the formula defining $ \Pi_{ \Phi_1, k}$ for $g \in L^2_{k}( \mathbb{C}, \Phi_1)$ (see Proposition \ref{prop_ecriture_integrale_Pi_Phi1}):
$$ \Pi_{ \Phi_1, k} g(z) = \int_{ \mathbb{C}} e^{-(1/ 4 \hbar)(z- \overline{w})^2} g(w) e^{-2 \Phi_1(w)/ \hbar} L(dw) .$$
Let $g \in \mathcal{G}_k$, by a simple computation, we notice that, for $(m, n) \in \mathbb{Z}^2$ and for $z \in \mathbb{C}$, we have:
$$ \left\lbrace
\begin{split}
& g(z+ 2 \pi m) = \left( u^k \right)^m g(z), \\
& g(z+in) = \left( v^k e^{-izk+kn/2} \right)^n g(z).
\end{split}
\right.$$
Then, for $g \in \mathcal{G}_k$, we can write an estimate of the integral defining $ \Pi_{ \Phi_1,k}$ as follows:
\begin{align*}
& \int_{ \mathbb{C}} \left| e^{-(1/ 4 \hbar)(z- \overline{w})^2} g(w) e^{-2 \Phi_1(w)/ \hbar} \right| L(dw) , \\
& = \sum_{m \in \mathbb{Z}} \int_{[0, 2 \pi] + i \mathbb{R}} \left| e^{-(1/ 4 \hbar)(z- (\overline{w + 2 \pi m}))^2} g(w + 2 \pi m) e^{-2 \Phi_1(w + 2 \pi m)/ \hbar} \right| L(dw) , \\
& = \sum_{m \in \mathbb{Z}} \int_{[0, 2 \pi] + i \mathbb{R}} \left| e^{-(1/ 4 \hbar)(z- \overline{w} - 2 \pi m)^2} \left( u^k \right)^m g(w) e^{-2 \Phi_1(w)/ \hbar} \right| L(dw), \\
& = \sum_{(m, n) \in \mathbb{Z}^2} \int_{[0, 2 \pi] + i [0, 1]} \left| e^{-(1/ 4 \hbar)(z- (\overline{w +in}) - 2 \pi m)^2} \left( u^k \right)^m g(w+in) e^{-2 \Phi_1(w+in)/ \hbar} \right| L(dw), \\
& = \sum_{(m, n) \in \mathbb{Z}^2} \int_{[0, 2 \pi] + i [0, 1]} \left| e^{-(1/ 4 \hbar)(z- \overline{w} +in - 2 \pi m)^2} \left( u^k \right)^m \left( v^k e^{-iwk+kn/2} \right)^n g(w) \right. \\
& \left. e^{-2 \Phi_1(w)/ \hbar} e^{-n^2/ \hbar} e^{-2n \Im(w) / \hbar} \right| L(dw), \\
& = \sum_{(m, n) \in \mathbb{Z}^2} \int_{[0, 2 \pi] + i [0, 1]} \left| e^{-(1/ 4 \hbar)(z- \overline{w} +in - 2 \pi m)^2} \left( u^k \right)^m \left( v^k \right)^n e^{-kn^2/2} e^{-ink \overline{w}} g(w) e^{-2 \Phi_1(w)/ \hbar} \right| L(dw).
\end{align*}
For all $z \in \mathbb{C}$, we have:
\begin{align*}
& \int_{[0, 2 \pi] + i [0, 1]} \left| e^{-(1/ 4 \hbar)(z- \overline{w} +in - 2 \pi m)^2} \left( u^k \right)^m \left( v^k \right)^n e^{-kn^2/2} e^{-ink \overline{w}} g(w) e^{-2 \Phi_1(w)/ \hbar} \right| L(dw) \\
& = \int_{[0, 2 \pi] + i [0, 1]} \left| e^{-(1/ 4 \hbar)(z- \overline{w} +in - 2 \pi m)^2} e^{-kn^2/2} e^{-ink \overline{w}} g(w) e^{-2 \Phi_1(w)/ \hbar} \right| L(dw) \\
& \leq \| g \|_{ \mathcal{G}_k}^2 \int_{[0, 2 \pi] + i [0, 1]} \left| e^{-(1/ 4 \hbar)(z- \overline{w} +in - 2 \pi m)^2} e^{-kn^2/2} e^{-ink \overline{w}} \right|^2 e^{-2 \Phi_1(w)/ \hbar} L(dw) \\
& \text{using Cauchy-Schwartz in $L^2( [0, 2 \pi]+ i [0,1], e^{-2 \Phi_1(z)/ \hbar} L(dz))$}, \\
& = \| g \|_{ \mathcal{G}_k}^2 \int_{[0, 2 \pi] + i [0, 1]} \left| e^{-(1/ 4 \hbar)(z- \overline{w} +in - 2 \pi m)^2} e^{-ink \overline{w}} \right|^2 e^{-kn^2} e^{-2 \Phi_1(w)/ \hbar} L(dw), \\
& \leq C \| g \|_{ \mathcal{G}_k}^2 e^{- k \Re (z^2)/2} e^{-2 k \pi^2 m^2} e^{-k n^2/2} \max \left( e^{2 \pi m k \Re(z)}, e^{2 \pi (m+1) k \Re(z)}\right) \max \left( e^{n k \Im(z)}, e^{(n+1) k \Im(z)}\right), 
\end{align*}
where $C$ is a constant independent of $k$.
We recognize the general term of a convergent series in $m$ and $n$, thus according to Fubini's theorem, $ \Pi_{\Phi_1, k}$ is well-defined on $ \mathcal{G}_k$ by the following formula, for $g \in \mathcal{G}_k$:
\begin{align*}
& \Pi_{ \Phi_1, k} g(z) = \int_{ \mathbb{C}} e^{-(1/ 4 \hbar)(z- \overline{w})^2} g(w) e^{-2 \Phi_1(w)/ \hbar} L(dw), \\
& = \sum_{(m, n) \in \mathbb{Z}^2} \int_{[0, 2 \pi] + i [0, 1]} e^{-(1/ 4 \hbar)(z- \overline{w} +in - 2 \pi m)^2} \left( u^k \right)^m \left( v^k \right)^n e^{-kn^2/2} e^{-ink \overline{w}} g(w) e^{-2 \Phi_1(w)/ \hbar} L(dw).
\end{align*}
Now, since the range of $ \Pi_{\Phi_1,k}$ consists of holomorphic functions, then for $g \in \mathcal{G}_k$, $ \Pi_{ \Phi_1,k} g \in \Hol( \mathbb{C})$. Then, via the change of variables $w \longmapsto w + 2 \pi$, we prove with simple integral equalities that, for $g \in \mathcal{G}_k$, we have $\Pi_{ \Phi_1, k} g (z+ 2 \pi) =  u^k \Pi_{ \Phi_1, k} g(z)$ and using the change of variables $w \longmapsto w +i$, we also obtain that, for $g \in \mathcal{G}_k$,
$\Pi_{ \Phi_1, k} g(z+i) =v^k e^{-ikz+k/2} \Pi_{ \Phi_1, k} g(z)$. Finally, we recall that $ \Pi_{ \Phi_1, k} = \id$ on $ \mathcal{H}_k$ comes from Proposition \ref{prop_TphiTphi*=id_sur_Hk_et_Tphi*Tphi=id_sur_Lk}.
\end{proof}

We can now define the action of a Berezin-Toeplitz operator of the complex plane on the space $ \mathcal{H}_k$.

\begin{prop} \label{prop_action_toeplitz_sur_C_sur_Hk}
Let $f_k \in \mathcal{C}^{ \infty}_k( \mathbb{C})$ be a function such that, for $ z \in \mathbb{C}$, we have:
$$ f_k(z +2 \pi) = f_k(z) = f_k(z+i).$$
Let $T_{f_k} = (T_k)_{k \geq 1}$ be the Berezin-Toeplitz operator of the complex plane of symbol $f_k$.
Then, for $k \geq 1$ and for $v \in \mathcal{H}_k$, we have:
$$ T_k v = \Pi_{ \Phi_1, k} (f_k v ),$$
where $ \Pi_{ \Phi_1, k}$ is seen as the operator which sends $ \mathcal{G}_k$ on $ \mathcal{H}_k$ (see Proposition \ref{prop_Pi_Phi_bien_defini_sur_Hk}).
\end{prop}

\begin{proof}
According to Proposition \ref{prop_toeplitz_sur_C_defini_sur_S'}, for $v \in \mathcal{H}_k$, for $ u \in \mathfrak{S}( \mathbb{C})$ and for $k \geq 1$, we have:
\begin{align*}
\langle T_k v, u \rangle_{ \mathfrak{S}', \mathfrak{S}} & = \langle v, \tilde{T}_k u \rangle_{ \mathfrak{S}', \mathfrak{S}}, \\
& = \langle g, \tilde{T}_k u \rangle \quad \text{with $g \in \mathfrak{S}^{-l}( \mathbb{C})$ for $l \in \mathbb{N}$ according to Proposition 
\ref{prop_ecriture_elements_de_SS'(C)}}, \\
& = \langle g, \Pi_{ \Phi_1, k} M_{\overline{f}_k} \Pi_{ \Phi_1, k} u \rangle \quad \text{by definition of $ \tilde{T}_k$ on $ \mathfrak{S}( \mathbb{C})$}, \\
& = \langle \Pi_{ \Phi_1, k} M_{\overline{ \overline{f}}_k} \Pi_{ \Phi_1, k} g, u \rangle \quad \text{since $ \Pi_{ \Phi_1,k}^* = \Pi_{ \Phi_1,k}$}.
\end{align*}
Thus, for $v \in \mathcal{H}_k$ and for $k \geq 1$, we obtain, according to Proposition \ref{prop_Pi_Phi_bien_defini_sur_Hk}:
$$ T_k v = \Pi_{ \Phi_1, k} M_{f_k} \Pi_{ \Phi_1, k} v = \Pi_{ \Phi_1, k} (f_k v ).$$
\end{proof}

We deduce from Proposition \ref{prop_action_toeplitz_sur_C_sur_Hk} and Proposition \ref{prop_Pi_Phi_bien_defini_sur_Hk}, a result which relates a Berezin-Toeplitz operator of the complex plane and a Berezin-Toeplitz operator of the torus. To our knowledge, this fact is new in the literature and it is also fundamental to prove Theorem \ref{theoA}. 

\begin{prop} \label{prop_toeplitz_sur_C_egal_toeplitz_sur_Hk}
Let $f_k \in \mathcal{C}^{ \infty}_k( \mathbb{C})$ be a function such that, for $ z \in \mathbb{C}$, we have:
$$ f_k(z +2 \pi) = f_k(z) = f_k(z+i).$$
Let $T_{f_k}^{ \mathbb{C}} = ( T_k^{ \mathbb{C}})_{k \geq 1}$ be the Berezin-Toeplitz operator of the complex plane of symbol $ f_k$ and let $T_{f_k}^{ \mathbb{T}^2} = (T_k^{ \mathbb{T}^2})_{k \geq 1}$ be the Berezin-Toeplitz operator of the torus of symbol $f_k$. Then, for $k \geq 1$, we have:
$$ T_k^{ \mathbb{C}} = T_k^{ \mathbb{T}^2} + \mathcal{O}(k^{- \infty}) \quad \text{on $ \mathcal{H}_k$}.$$
Consequently, a Berezin-Toeplitz operator of the complex plane whose symbol is periodic coincides with a Berezin-Toeplitz operator of the torus.
\end{prop}

\subsubsection{Berezin-Toeplitz quantization and complex Weyl quantization of the torus}
Finally, we are able to establish and to prove the following proposition which corresponds to Theorem \ref{theoA}.

\begin{prop}[Theorem \ref{theoA}]
Let $ f_k \in \mathcal{C}^{ \infty}_k( \mathbb{R}^2)$ be a function such that, for $(x, y) \in \mathbb{R}^2$, we have:
$$ f_k(x + 2 \pi, y ) = f_k(x, y) = f_k(x, y+1) .$$
Let $T_{f_k} = ( T_k)_{k \geq 1}$ be the Berezin-Toeplitz operator of the torus of symbol $ f_k$. Then, for $k \geq 1$, we have:
$$ T_k = \Op^w_{ \Phi_1}(b_{ \hbar}) + \mathcal{O}(k^{- \infty}) \quad \text{on $ \mathcal{H}_k$},$$
where $b_{ \hbar} \in \mathcal{C}^{ \infty}_{ \hbar}( \Lambda_{ \Phi_1})$ is defined by the following formula, for $z \in \Lambda_{ \Phi_1} \simeq \mathbb{C}$:
\begin{equation} \label{eq_formule_fk_en_fonction_ah}
b_{ \hbar}(z) = \exp \left( \dfrac{1}{k} \partial_z \partial_{ \overline{z}} \right) (f_k(z)).
\end{equation}
Besides, $b_{ \hbar}$ satisfies the following periodicity conditions, for $(z,w) \in \Lambda_{ \Phi_1} $:
$$ b_{ \hbar}(z+ 2 \pi,w) = b_{ \hbar}(z,w) = b_{ \hbar}(z+i,w-1).$$
\end{prop}

\begin{proof}
Since in particular $f_k \in S( \mathbb{C})$, if we denote by $T_{f_k}^{ \mathbb{C}} = (T_k^{ \mathbb{C}})_{k \geq 1}$ the Berezin-Toeplitz operator of the complex plane of symbol $f_k$, then according to Proposition  \ref{prop_Toeplitz=pseudo_H(Phi1)}, there exists $b_{ \hbar} \in S( \Lambda_{ \Phi_1})$ such that, for $k \geq 1$, we have:
$$ T_k^{ \mathbb{C}} = \Op^w_{ \Phi_1}(b_{ \hbar}) \quad \text{on $H_{ \hbar}( \mathbb{C}, \Phi_1)$},$$
where $b_{ \hbar}$ is given by the following formula, for $z \in \Lambda_{ \Phi_1} \simeq \mathbb{C}$:
$$ b_{ \hbar}(z) = \exp \left( \dfrac{1}{k} \partial_z \partial_{ \overline{z}} \right) (f_k(z)).$$
Since $ \mathfrak{S}( \mathbb{C})$ is included into $H_{ \hbar}( \mathbb{C}, \Phi_1)$ (see Remark \ref{rema_S(C)_inclus_dans_H(C,Phi1)}) then, by restriction, we obtain:
$$ T_k^{ \mathbb{C}} = \Op^w_{ \Phi_1}(b_{ \hbar}) \quad \text{on $\mathfrak{S}( \mathbb{C})$}.$$
By duality, we have:
$$ T_k^{ \mathbb{C}} = \Op^w_{ \Phi_1}(b_{ \hbar}) \quad \text{on $\mathfrak{S}'( \mathbb{C})$}.$$
Since $ \mathcal{H}_k \subset \mathfrak{S}'( \mathbb{C})$, we obtain:
$$ T_k^{ \mathbb{C}} = \Op^w_{ \Phi_1}(b_{ \hbar}) \quad \text{on $\mathcal{H}_k$}.$$
Notice that the periodicity conditions on $f_k$ and Equation \eqref{eq_formule_fk_en_fonction_ah} give the periodicity conditions on $b_{ \hbar}$, consequently, $ \Op^w_{ \Phi_1}(b_{ \hbar})$ is well-defined on $ \mathcal{H}_k$. \\
Finally, according to Proposition \ref{prop_toeplitz_sur_C_egal_toeplitz_sur_Hk}, if we denote by $ T_{f_k}^{ \mathbb{T}^2} = (T_k^{ \mathbb{T}^2})_{k \geq 1}$ the Berezin-Toeplitz operator of the torus of symbol $f_k$, we have, for $k \geq 1$:
$$ T_k^{ \mathbb{C}} = T_k^{ \mathbb{T}^2} + \mathcal{O}(k^{- \infty}) \quad \text{on $ \mathcal{H}_k$}.$$
This concludes the proof.
\end{proof}

Thanks to Theorem \ref{theoA} and Proposition \ref{prop_Toeplitz=pseudo_H(Phi1)}, we deduce the following corollary.

\begin{cora} \label{coro}
Let $ f_k \in \mathcal{C}^{ \infty}_k( \mathbb{R}^2)$ be a function such that, for $(x, y) \in \mathbb{R}^2$, we have:
$$ f_k(x + 2 \pi, y ) = f_k(x, y) = f_k(x, y+1) .$$
Let $T_{f_k} = ( T_k)_{k \geq 1}$ be the Berezin-Toeplitz operator of the torus of symbol $ f_k$. Then, for $k \geq 1$, we have:
$$ T_{ \phi_1}^* T_k T_{ \phi_1} = \Op^w(a_{ \hbar}) + \mathcal{O}(\hbar^{\infty}) \quad \text{on $ L^2( \mathbb{R})$},$$
where $a_{ \hbar} \in \mathcal{C}^{ \infty}_{ \hbar}(\mathbb{R}^2)$ is defined by the following formula:
$$ a_{ \hbar} = b_{ \hbar} \circ \kappa_{ \phi_1},$$
where $b_{ \hbar} \in \mathcal{C}^{ \infty}_{ \hbar}( \Lambda_{ \Phi_1})$ is defined by Equation \eqref{eq_formule_fk_en_fonction_ah}. Besides, $a_{ \hbar}$ satisfies the following periodicity conditions, for $(x, y) \in \mathbb{R}^2$:
$$ a_{ \hbar}(x+ 2 \pi, y) = a_{ \hbar}(x, y) = a_{ \hbar}(x, y+1).$$
\end{cora}

\newpage

\renewcommand{\thesection}{\Alph{section}}
\setcounter{section}{0}

\nocite{*}
\bibliographystyle{amsalpha}
\bibliography{biblio}

\end{document}